\documentclass[times,sort&compress,3p]{elsarticle}
\journal{Journal of Multivariate Analysis}
\usepackage[labelfont=bf]{caption}

\renewcommand{\figurename}{Fig.}

\usepackage{amsmath,amsfonts,amssymb,amsthm,booktabs,color,epsfig,graphicx,hyperref,url,multirow,latexsym,bm,srcltx,natbib,wasysym}

\theoremstyle{plain}

\newtheorem{proposition}{Proposition}
\newtheorem{lemma}{Lemma}

\theoremstyle{definition}

\numberwithin{equation}{section}

\begin{document}

\begin{frontmatter}

\title{
Mat\'ern and Generalized Wendland correlation models that parameterize hole effect, smoothness, and support}

\author[A1]{Xavier Emery}
\author[A2]{Moreno Bevilacqua\corref{mycorrespondingauthor}}
\author[A3]{Emilio Porcu}

\address[A1]{Department of Mining Engineering, Universidad de Chile, Santiago, Chile\\
Advanced Mining Technology Center, Universidad de Chile, Santiago, Chile}
\address[A2]{Facultad de Ingenieria y Ciencias, Universidad Adolfo Iba\~nez, Vi\~na del Mar, Chile\\
Dipartimento di Scienze Ambientali, Informatica e Statistica, Ca’ Foscari University of Venice, Italy}
\address[A3]{Department of Mathematics and Center for Biotechnology (BTC), Khalifa University, Abu Dhabi, United Arab Emirates\\
ADIA Lab, Abu Dhabi, United Arab Emirates}

\cortext[mycorrespondingauthor]{Corresponding author. Email address: moreno.bevilacqua@uai.cl (M. Bevilacqua)\url{}}

\begin{abstract}
A huge literature in statistics and machine learning is devoted to parametric families of correlation functions, where the correlation parameters are used to understand the properties of an associated spatial random process in terms of smoothness and global or compact support. However, most of current parametric correlation functions attain only non-negative values. This work provides two new families {of correlation functions that can have some negative values} (aka hole effects), along with smoothness, and global or compact support. They generalize the celebrated Mat\'ern and Generalized Wendland models, respectively, which are {obtained} as special cases. A link between the two new families is also established, showing that a specific reparameterization of the latter includes the former as a special limit case. Their performance in terms of estimation accuracy and goodness of best linear unbiased prediction is illustrated through synthetic and real data.

\end{abstract}


\end{frontmatter}

\section{Introduction}

The data science revolution provides a collection of research challenges and triggers an increasing appetite for new stochastic models that allow {describing} complex realities. In this context, covariance functions have proved useful to describe and analyze a wide portfolio of real-life data in spatial statistics,  machine learning  and related disciplines. In this manuscript, we focus on Gaussian random fields, for which covariance functions are crucial to modeling, estimation, prediction, {and simulation}.

\subsection{What Should a Covariance Model Describe?}
For a Gaussian random field in $\mathbb{R}^d$, for $d$ a positive integer, it is customary to assume that the covariance function is stationary and isotropic. That is, the covariance between observations at any two points depends solely on the distance between the points. The paper works under this assumption, which simplifies the discussion considerably. However, isotropic covariance models represent the building blocks for more complex scenarios such as anisotropy  or nonstationarity, to mention a few.

For a given spatial data set, one can be interested in understanding
\begin{enumerate}
\item The \emph{smoothness} of the realizations (sample paths) of the underlying Gaussian random field, e.g., mean square continuity, mean square differentiability, and fractal dimensions. This aspect covers a central part of the literature, starting with \cite{Adler:1981} and following with \cite{stein-book} and \cite{Chiles2012} as classical textbooks on this subject. Continuity, differentiability and fractal dimensions of a Gaussian random field are in one-to-one correspondence with the local behavior of the covariance function (read: continuity and differentiability of some given order at the origin).

\item The \emph{correlation range}, defined as the distance beyond which the spatial dependence is identically equal to zero. The covariance function is said to be compactly or globally supported, depending on whether this range is finite or not. The compact support is a desirable feature from a computational viewpoint, since sparse matrix algorithms  \citep{sparsedavis,Scott2023}  can be exploited to speed-up the computation associated with estimation, prediction and/or simulation of Gaussian random fields.

\item Positive and negative dependencies. The occurrence of negative dependencies, a phenomenon known as a {\em hole effect} in geostatistics \citep{Chiles2012}, is of interest in various disciplines of the natural sciences and engineering  (see \cite{alegria2023} and references therein). {For instance, in landscape and population ecology, hole effects arise due to local interaction processes \citep{bellier2010}. In air quality monitoring, they can be an outcome of dynamic atmospheric conditions and government policies \citep{alegria2023}. In geology, sedimentary and diagenetic processes can explain alternating decreases and increases of rock properties (porosity, resistivity, photoelectric absorption capability, etc.), which translate into hole effects in their spatial correlation structure \citep{Lefranc2008, parra2013, Matonti}. A similar phenomenon occurs in precision farming, with the alternation of compacted and uncompacted soils due to tillage \citep{sanmartin}. However, the modeling of empirical covariance functions that exhibit hole effects is often arduous, as most isotropic models} used in applications can only attain strictly positive values (for globally supported models) or non-negative values (for models with compact support), and only a few models oscillate between positive and negative values.

\end{enumerate}

\subsection{Parameterization is All You Need}

To date, there is a rich catalog of parametric families of covariance functions \citep{Chiles2012}. For some of them, the parameterization of the local behavior and the local or global support is possible. In particular, the Mat{\'e}rn  \citep[][and references therein]{porcu2024matern} and Generalized Wendland \citep[][and references therein]{BFFP} families do the job. Both families allow {continuously parameterizing} the mean square properties of the associated Gaussian random field. The Matérn covariance is globally supported and attains strictly positive values, while the Generalized Wendland covariance is compactly supported and has a parameter that determines the correlation range. Having parametric families of covariances that identify these {aspects} has considerable advantages:
\begin{itemize}
\item[(a)] model interpretability: each parameter is associated with a feature of interest for the underlying random field;
\item[(b)] feasible estimation techniques: there exists a well established literature about the estimation of the parameters associated with both the Mat\'ern and Generalized Wendland families, including an asymptotic assessment of the estimation accuracy under different asymptotic schemes \citep{mardia1984maximum,Shaby:Kaufmann:2013,BFFP};
\item[(c)] prediction accuracy under a specific asymptotic scheme can be quantified according to a combination of the parameters indexing these families  \citep{stein-book}.
\end{itemize}

\subsection{Challenge and Contribution}

While smoothness and support have been extensively studied, the literature on hole effect models is scarce. Elegant arguments from \cite{schoenberg1938metric2} allow {deducing} lower bounds for isotropic covariance models that attain negative values in some part of their domain. However, the models currently in use mostly consist of damped periodic functions, such as the Bessel-J {covariance that is differentiable at the origin and globally supported. Compactly-supported models with} finitely many oscillations are still given little consideration {and generally do not allow parameterizing smoothness;} the reader is referred to \cite{jkjk} for a state-of-the art review and a comprehensive survey of applications. A motive for this disaffection is the shortage of versatile parametric families that, in addition to the hole effect, keep the original features of parameterizing {both} smoothness and support.

Two solutions to this issue are proposed. Starting from the Mat{\'e}rn and Generalized Wendland models, we derive two new families that additionally allow {indexing} the hole effect for the associated Gaussian random field.
The key tool to obtain the proposed solutions is the iterative application of the turning bands operator to the standard Mat{\'e}rn and Generalized Wendland covariance models. As outlined in \cite{gneiting2002compactly}, such an operator preserves the local behavior of the covariance function at the origin and allows {attaining} negative value at the same time.

The computation of the proposed new models heavily depends on the evaluation of some special functions. However, we provide closed-form solutions for important special cases, which makes them attractive to practitioners. Finally, we establish a connection  between the proposed models, by showing that a reparameterization of the hole effect Generalized Wendland model includes the hole effect Mat\'ern model as a special limit case. Both proposed models have been implemented in the \texttt{GeoModels} package \citep{Bevilacqua:2018aa} for the open-source R statistical environment.

The remainder of the paper is organized as follows. Section \ref{sec2} reviews the celebrated Mat\'ern and Generalized Wendland models. Section \ref{secti3} includes the main theoretical results, in particular it introduces the hole effect Mat\'ern and hole effect Generalized Wendland models (Propositions \ref{matgeneral} and \ref{genW}, respectively) and the connection between the former and a reparameterized version of the latter (Proposition \ref{wend2mat}). In Section \ref{sec:estim}, we report a small simulation study that explores the finite sample properties of the maximum likelihood method when  estimating the parameters of the reparameterized hole effect Generalized Wendland model. In Section \ref{sec:real}, we apply this model to the analysis of soil data. Concluding remarks are consigned in Section \ref{sec:concl}. Additional background material, technical lemmas and proofs are deferred to {the Appendix} (Section \ref{suppmaterial}).

\section{Background}
\label{sec2}

This section exposes the necessary background material and notation. Throughout, $d$ is a positive integer and $a$ a positive real number. Table \ref{tab:special function0} summarizes the set of ordinary and special functions used in this paper, the definition of which will therefore be omitted in the sequel.

\begin{table}[ht!]
    \caption{Functions used in the paper.}
    \label{tab:special function0}
    \begin{tabular}{ c  l l }
        \hline
        Notation & Function name & Parameters \\
        \hline
        $\|\cdot\|_d$ & Euclidean norm in $\mathbb{R}^d$ &\\
        $(\cdot)_+$ & Positive part function &\\
        $(\cdot)_n$ & Pochhammer symbol (rising factorial)& $n \in \mathbb{N}$\\
        $\Gamma$ & Gamma function &\\
        $J_{\nu}$ & Bessel function of the first kind & $\nu \in \mathbb{R}$\\
        ${\cal K}_{\nu}$ & Modified Bessel function of the second kind & $\nu \in \mathbb{R}$\\
        ${}_2F_1({\alpha,\beta;\gamma};\cdot)$ & Gauss hypergeometric function & $\alpha,\beta,\gamma \in \mathbb{R}$\\       ${}_pF_q({\boldsymbol{\beta};\boldsymbol{\gamma}};\cdot)$ & Generalized hypergeometric function & $p,q \in \mathbb{N}$, $\boldsymbol{\beta} \in \mathbb{R}^p, \boldsymbol{\gamma}\in \mathbb{R}^q$\\
        \hline
    \end{tabular}
\end{table}

\subsection{Isotropic Correlation Functions and their Spectral Representations}
\label{sub21}

\label{spectralrep}
For a given covariance function associated with a Gaussian random field, the correlation is defined as the ratio between (a) the covariance at two different points and (b) the product of the standard deviations at the two points. Hence, the correlation is a rescaled covariance function.

A real-valued zero-mean Gaussian random field $\{{Z}(\boldsymbol{x}): \boldsymbol{x} \in \mathbb{R}^{d}\}$ is second-order stationary and isotropic if, for any $\boldsymbol{x}$ and $\boldsymbol{x}^\prime$ in $\mathbb{R}^{d}$, the correlation $K(\boldsymbol{x},\boldsymbol{x}^\prime)$ between ${Z}(\boldsymbol{x})$ and ${Z}(\boldsymbol{x}^\prime)$ exists and only depends on the separation distance $\| \boldsymbol{x}-\boldsymbol{x}^\prime \|_d$:
\begin{equation}
\label{eq:stationarycov_Rd}
K(\boldsymbol{x},\boldsymbol{x}^\prime): = {{\rm corr}}({Z}(\boldsymbol{x}), {Z}(\boldsymbol{x}^{\prime})) = {C}\left(\| \boldsymbol{x}-\boldsymbol{x}^\prime \|_d \right), \quad \boldsymbol{x},\boldsymbol{x}^{\prime} \in \mathbb{R}^{d}.
\end{equation}

Correlation functions are positive semidefinite. For the function $K$ as per Equation (\ref{eq:stationarycov_Rd}), this implies that the matrix $[C(\| \boldsymbol{x}_i - \boldsymbol{x}_j \|_d)]_{i,j=1}^p$ is symmetric positive semidefinite for any positive integer $p$ and any choice of $\boldsymbol{x}_1, \ldots, \boldsymbol{x}_p \in \mathbb{R}^d$. We refer to $C$ as the $d$-radial correlation function of the random field $Z$, as a shorthand to the radial part of the correlation function $K$ in $\mathbb{R}^d \times \mathbb{R}^d$.

We denote $\Phi_{d}$ the class of continuous mappings ${C}:[0,+\infty) \to \mathbb{R}$  with $C(0)=1$ such that (\ref{eq:stationarycov_Rd}) is true for a second-order stationary isotropic Gaussian random field in $\mathbb{R}^{d}$. The following strict inclusion relations hold: $\Phi_1 \supset \Phi_2 \supset \ldots \supset \Phi_{\infty}:=\cap_{n=1}^{+\infty} \Phi_n$.

Elements $C$ of the class $\Phi_d$ such that $C(\| \cdot \|_d)$ is absolutely integrable in $\mathbb{R}^d$ admit the following Fourier-Hankel representation \citep{Chiles2012}:
\begin{equation}
    \label{fourier1}
    C(h) = {(2\pi)^{{d/2}}} h^{1-{d/2}} \int_0^{+\infty} u^{{d/2}} J_{{d/2}-1}(u h) \widehat{C}_d(u) {\rm d}u, \qquad h > 0,
\end{equation}
with
\begin{equation}
    \label{fourier2}
    \widehat{C}_d(u) = \frac{1}{(2\pi)^{{d/2}}} u^{1-{d/2}} \int_0^{+\infty} h^{{d/2}} J_{{d/2}-1}(u h) C(h) {\rm d}h, \qquad u > 0,
\end{equation}
where $\widehat{C}_d: (0,+\infty) \to [0,+\infty)$, which will be referred to as the $d$-radial spectral density of $C$ or of $K$, is a mapping such that $\widehat{C}_d(\| \cdot \|_d)$ is a probability density on $\mathbb{R}^d$. Despite the similarity between the direct and inverse Fourier-Hankel transforms (\ref{fourier1}) and (\ref{fourier2}), the functions $C$ and $\widehat{C}_d$ do not play symmetrical roles: $\widehat{C}_d$ is non-negative, but $C$ can take negative values. The amplitude of the negative values decreases with the space dimension $d$, with $C$ being lower-bounded by $-1/d$ \citep[p. 13]{Matern}.

\subsection{The Mat\'ern Parametric Family of Correlation Models}

The Mat\'ern model is a two-parameter globally  supported correlation function,  that allows for a continuous parameterization of the smoothness of the underlying Gaussian random field.
 It is   defined as \citep[3.3.10]{Matern}:
\begin{equation}\label{Matern}
{\cal M}_{a,\xi} (h) = \begin{cases}
1 \qquad \qquad \qquad \qquad \text{if } h = 0,\\
\frac{2^{1-\xi}}{\Gamma(\xi)} \left ( \frac{h}{a} \right )^{\xi} {\cal K}_{\xi}\left ( \frac{h}{a} \right ), \qquad \text{if } h > 0,
\end{cases}
\end{equation}
where $a,\xi >0$ are necessary and sufficient conditions for $ {\cal M}_{a,\xi}  \in \Phi_{\infty}$. The associated $d$-radial spectral density, $\widehat{{\cal M}}_{a,\xi,d}$, is given by:
\begin{equation}
\label{stein1}
\widehat{{\cal M}}_{a,\xi,d}(u)= \frac{\Gamma(\xi+\frac{d}{2})}{\pi^{{d/2}} \Gamma(\xi)}
\frac{a^d}{(1+a^2u^2)^{\xi+{d/2}}}
, \qquad u \ge 0.
\end{equation}

The Mat\'ern model is globally supported, that is ${\cal M}_{a,\xi}>0$. The importance of this  model stems from the parameter $\xi$ that controls the differentiability (in the mean square sense) of the associated Gaussian random field and of its sample paths. Specifically, for any integer $\ell=0,1,\ldots$, the sample paths of a Gaussian random field with correlation function ${\cal M}_{a,\xi}$ are $\ell$-times differentiable, in any direction, if and only if $\xi > \ell$. When $\xi=\ell+ 1/2$, the Mat\'ern  correlation simplifies into the product of an exponential correlation with a polynomial of degree $\ell$:
$${\cal M}_{a,\ell+1/2}(h)= \exp(-h/a) \sum_{i=0}^{\ell} \frac{(\ell+i)!}{2 \ell !} \binom{l}{i} (2h/a)^{\ell-i} \quad \ell=0,1,\ldots.$$

\subsection{The Generalized Wendland Family of Correlation Models}

The Generalized Wendland model \citep{gneiting2002compactly} is a three-parameter  compactly supported correlation function that allows for a continuous parameterization of the smoothness of the underlying Gaussian random field. For $\xi > -\frac{1}{2}$ and a compact support parameter $a > 0$, this model and its associated $d$-radial spectral density are defined as \citep{Chernih, BFFP}:
\begin{equation} \label{WG4*}
{\cal GW}_{a,\xi,\nu}(h) = \begin{cases}  \frac{\Gamma(\xi)\Gamma(2\xi+\nu+1)}{\Gamma(2\xi)\Gamma(\xi+\nu+1)2^{\nu+1}} \left( 1- \frac{h^2}{a^2} \right)^{\xi+\nu} {{}_2F_1\left(\frac{\nu}{2},\frac{\nu+1}{2};{\xi+\nu+1};{1-\frac{h^2}{a^2}}\right)},& 0 \leq h < a,\\
     0,&h \geq a, \end{cases}
\end{equation}
and
\begin{equation}
\label{llkk}
\widehat{{\cal GW}}_{a,\xi,\nu,d}(u)=\frac{a^{d}\Gamma(\xi+\frac{d+1}{2}) \Gamma(2\xi+\nu+1)}{\pi^{{d/2}} \Gamma(\xi+\frac{1}{2}) \Gamma(2\xi+\nu+1+d)}
{{}_1F_2\left(\frac{d+1}{2}+\xi;\frac{d+1+\nu}{2}+\xi,\frac{d+\nu}{2}+1+\xi;-\frac{a^2 u^{2}}{4}\right)}, \quad  u \geq 0.
\end{equation}

A necessary and sufficient condition for ${\cal GW}_{a,\xi,\nu}$ to belong to $\Phi_d$ is \citep{bevi2024}
\begin{equation}
\label{bevcondition}
    \nu \geq \nu_{\min}(\xi,d) :=
    \begin{cases}
    \frac{\sqrt{8\xi+9}-1}{2} \text{ if $d = 1$ and $-\frac{1}{2}<\xi<0$}\\
    \xi+ \frac{d+1}{2} \text{ otherwise}.
    \end{cases}
\end{equation}

It should be stressed that for a given smoothness parameter $\xi$ and compact support parameter $a$, $ \nu$  allows {parameterizing} the shape of the correlation function.
For $\ell=1,2\ldots$, the sample paths of a Gaussian random field with correlation function ${\cal GW}_{a,\xi,\nu}$ are $\ell$ times differentiable, in any direction, if and only if $\xi>\ell-\frac{1}{2}$, while for $-\frac{1}{2}< \xi < \frac{1}{2}$ they are not differentiable.

Similarly to the Mat\'ern model, when $\xi = \ell$ is a non-negative integer, the Generalized Wendland  correlation simplifies into the product of an Askey (truncated power) correlation with a polynomial $P_\ell$ of degree $\ell$
\citep{bevi2024}:
\begin{equation}\label{ppoo}
{\cal GW}_{a,\ell,\nu}(h)=\left(1-\frac{h}{a} \right)^{\nu+\ell}_+P_\ell(h;\nu,a), \quad \ell=0,1,2,\ldots
\end{equation}

The Mat\'ern and Generalized Wendland models have conceptual and mathematical connections. For a specific parameter setting, both models lead to equivalent Gaussian measures \citep{BFFP}. In addition, the Mat\'ern model is a special limit case of a reparameterization of the Generalized Wendland model \citep{bevilacqua2022unifying}:
\begin{equation}\label{ppoi}
\lim_{\nu\to\infty}  {\cal GW}_{\delta,\xi,\nu}(h)={\cal M}_{a,\xi+1/2}(h), \quad \xi>-\frac{1}{2},
\end{equation}
 with uniform convergence for $h>0$,
where $\delta=a(\Gamma(\nu+2\xi+1)/\Gamma(\nu))^{\frac{1}{1+2\xi}}$.

\section{Parameterizing Smoothness, Supports, and Hole Effects}
\label{secti3}

\subsection{A Mat{\'e}rn-Type Model that Parameterizes Hole Effects}

Our first proposal details a correlation model having the same characteristics as the Mat{\'e}rn model, with the additional feature of parameterizing the hole effect.

\begin{proposition}[hole effect Mat\'ern correlation model]
\label{matgeneral}
    For $k \in \mathbb{N}$, $a,\xi > 0$, define
\begin{equation}
\label{explicitmatGeneralized}
\begin{split}
  {\cal M}_{a,\xi,d,k}(h) &:=\sum_{q=0}^k \sum_{r=0}^{\max\{0,q-1\}} \sum_{s=0}^{q-r} \sum_{t=0}^{q-r-s} \left ( \frac{h}{a} \right )^{\xi+q-r-s} {\cal K}_{\xi+2t+r+s-q}\left ( \frac{h}{a} \right )\\
  &\times \frac{(-1)^{q-s} (q-r)! (q-r)_r (\xi+1-s)_{s} (k-q+1)_q (q)_r}{2^{\xi+2q-s-1} q! \, r! \, s! \, t! \, (q-r-s-t)! \, \Gamma(\xi) (\frac{d}{2})_q}, \quad h>0.
\end{split}
\end{equation}
Then, ${\cal M}_{a,\xi,d,k}$ belongs to $\Phi_d$, and its $d$-radial spectral density is
\begin{equation}
\label{stein2}
\widehat{{\cal M}}_{a,\xi,d,k}(u)= \frac{\Gamma(\frac{d}{2}) \Gamma(\xi+\frac{d}{2}+k)}{\pi^{{d/2}} \Gamma(\frac{d}{2}+k) \Gamma(\xi)} \frac{a^{d+2k} u^{2k}}{(1+a^2u^2)^{\xi+{d/2}+k}}
, \qquad u \ge 0.
\end{equation}
Furthermore, ${\cal M}_{a,\xi,d,0} = {\cal M}_{a,\xi}$, as given in (\ref{Matern}).
\end{proposition}

We term ${\cal M}_{a,\xi,d,k}$ the $(d,k)$-hole effect Mat{\'e}rn model or, briefly, the hole effect Mat\'ern model. For $k=0$, we attain the standard Mat{\'e}rn model that has no hole effect. When positive, $k$ is an additional discrete parameter describing increasing levels of negative correlations that are functions not only of $k$, but also of $d$. This is not surprising since, as outlined in Section \ref{sub21}, the permissible negative correlation has a lower bound that depends on $d$. The role of the other parameters $(a,\xi)$ is unchanged.

Note that (\ref{stein2}) is a particular case of spectral densities proposed by \cite{vecchia1985} (for $d=2$) and \cite{laga2017} (for $d \geq 1$), but none of these authors provides a closed-form expression of the associated covariance models, as we do with (\ref{explicitmatGeneralized}).

Alternative expressions of ${\cal M}_{a,\xi,d,k}$ in terms of special functions (hypergeometric, Bessel-$I$ and Meijer-$G$ functions) can be found in {Appendix}. However, our implementation in \texttt{GeoModels}  \citep{Bevilacqua:2018aa} uses (\ref{explicitmatGeneralized})
to compute  ${\cal M}_{a,\xi,d,k}$.

Similarly to the standard  Mat\'ern model, the expression in (\ref{explicitmatGeneralized}) simplifies when $\xi=n+\frac{1}{2}$ with $n \in \mathbb{N}$, avoiding the computation of the Bessel-$K$ function. Using (\ref{explicitMatern}) in {Appendix} and formula 8.468 in \cite{grad}, one finds
\begin{equation*}  \label{MK}
\begin{split}
{\cal M}_{a,n+{1/2},d,k}(h) &=  \exp\left(-\frac{h}{a}\right) \sum_{q=0}^k \sum_{r=0}^{\max\{0,q-1\}} \sum_{s=0}^{q-r} \sum_{t = 0}^{n} \left(\frac{h}{a}\right)^{n+q-r-s-t}\\
& \times \frac{\pi^{\frac{1}{2}} (q-r)! (k-q+1)_q (q)_r (q-r)_r (n - t +1)_{2t} (n-s-t+1)_s}{(-1)^{q-s} \, 2^{n+q+r+t} \, q! \, r! \, s! \, t! \, (q-r-s)! \, (\frac{d}{2})_q \Gamma(n+\frac{1}{2})}, \quad h \geq 0.
\end{split}
\end{equation*}
For instance when $n=0$ and $k=1$, this gives
$$ {\cal M}_{a,{1/2},d,1}(h)= \exp\left(-\frac{h}{a}\right) \left(1-\frac{h}{a d} \right), \quad h  \ge 0,$$
and when $n=0$ and $k=2$, this gives:
$$ {\cal M}_{a,{1/2},d,2}(h)=\exp\left(-\frac{h}{a}\right) \left(1-\frac{h(2d+3)}{a d (d+2)} +\frac{h^2}{a^2 d (d+2)} \right), \quad h  \ge 0.$$

An illustration is provided in Figure \ref{holemat} for $n=0$ and $n=1$. It is seen that, when increasing $k$, the hole effect also increases and the correlation is flattened at the same time.

\begin{figure}[h!]
\begin{center}
\begin{tabular}{cc}
\includegraphics[width=6.2cm,height=6.2cm]{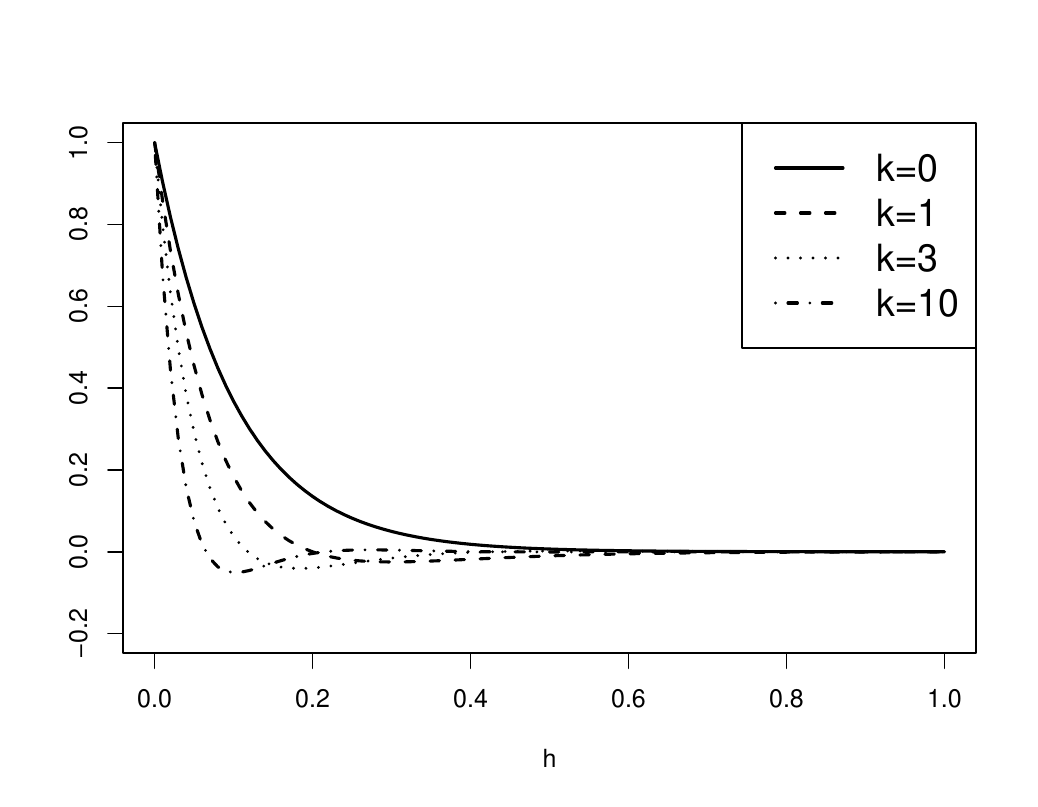}&\includegraphics[width=6.2cm,height=6.2cm]{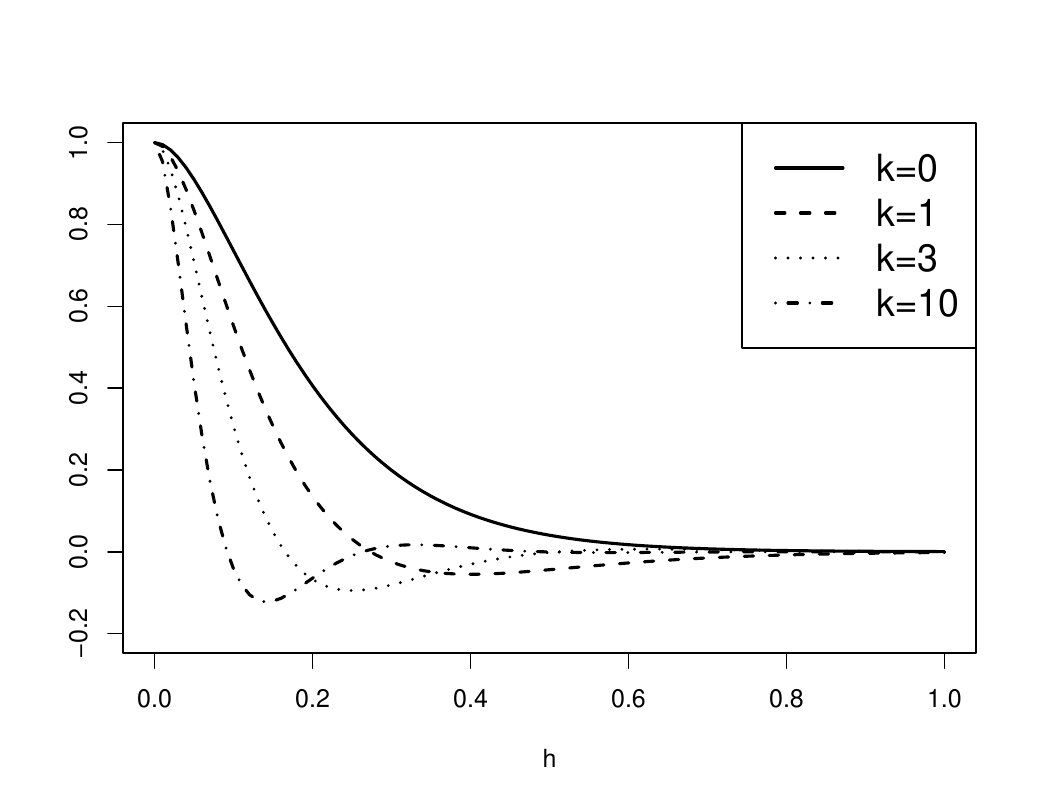}\\
\end{tabular}
\end{center}
\caption{Examples of the $(d,k)$-hole effect Mat\'ern model when $d=2$. Left: ${\cal M}_{0.05,0.5,2,k}(\cdot)$ for $k=0,1,3,10$ from top to bottom. Right: ${\cal M}_{0.05,1.5,2,k}(\cdot)$ for $k=0,1,3,10$ from top to bottom.}
 \label{holemat}
\end{figure}

\subsection{A Generalized Wendland-type Model that Parameterizes Hole Effects}

Our second proposal details a correlation function having the same characteristics as the Generalized Wendland   model, with the additional feature of parameterizing the hole effect.

\begin{proposition}[hole effect Generalized Wendland correlation model]
\label{genW}
For $a>0$, $\xi, \nu \in \mathbb{R}$ and $k \in \mathbb{N}$, the compactly supported mapping defined by
\begin{equation}
    \label{wendlandext0}
    \begin{split}
      {\cal GW}_{a,\xi,\nu,d,k}(h) &:=     
      {{}_3F_2\left(\frac{d}{2}+k,\frac{1-\nu}{2}-\xi,-\frac{\nu}{2}-\xi;\frac{1}{2}-\xi,\frac{d}{2};\frac{h^2}{a^2}\right)} \\&
    + {L_{\xi,\nu,a,d,k}}   \times \left(\frac{h}{a}\right)^{2\xi+1} 
    {{}_3F_2\left(\xi+\frac{d+1}{2}+k,1-\frac{\nu}{2},\frac{1-\nu}{2};\xi+\frac{3}{2},\xi+\frac{d+1}{2};\frac{h^2}{a^2}\right)}
    \end{split}
\end{equation} \text{ if $0 \leq h < a$,} and $0$ otherwise, with
$${L_{\xi,\nu,a,d,k} =  \frac{\Gamma(\xi+\frac{d+1}{2}+k) \Gamma(\xi+\frac{1+\nu}{2}) \Gamma(\xi+\frac{\nu}{2}+1) \Gamma(\frac{d}{2})\Gamma(-\xi-\frac{1}{2})}{\Gamma(\frac{d}{2}+k) \Gamma(\xi+\frac{1}{2})\Gamma(\frac{\nu}{2})\Gamma(\frac{\nu+1}{2}) \Gamma(\xi+\frac{d+1}{2})}}, $$
belongs to $\Phi_d$ when the following conditions hold:
\begin{enumerate}
    \item[(A)] $\xi > -\frac{1}{2}$
    \item[(B)] $\nu \geq \nu_{\min}(\xi,d+2k)$
    \item[(C)] $\xi+\frac{1}{2} \notin \mathbb{N}$.
\end{enumerate}
The $d$-radial spectral density is given by
\begin{equation}
\label{GWdensity}
\widehat{{\cal GW}}_{a,\xi,\nu,d,k} = {\widehat{L}}_{\xi,\nu,a,d,k}\, u^{2k} {{}_1F_2\left(\xi+\frac{d+1}{2}+k;\xi+\frac{d+\nu+1}{2}+k,\xi+\frac{d+\nu}{2}+k+1;-\frac{a^2 u^2}{4}\right)}, \quad u \geq 0,
\end{equation}
with $${ \widehat{L}}_{\xi,\nu,a,d,k} = \frac{a^{d+2k}\Gamma(\frac{d}{2})\Gamma(\xi+\frac{d+1}{2}+k) \Gamma(2\xi+\nu+1)}{\pi^{{d/2}} \Gamma(\frac{d}{2}+k) \Gamma(\xi+\frac{1}{2}) \Gamma(2\xi+\nu+1+d+2k)}.$$

Furthermore, ${\cal GW}_{a,\xi,\nu,d,0} = {\cal GW}_{a,\xi,\nu}$, as defined in (\ref{WG4*}). If condition (C) does not hold, ${\cal GW}_{a,\xi,\nu,d,k}$ can still be defined by (\ref{WG4*}) when $k=0$, or by continuation of (\ref{wendlandext0}) when $k \geq 1$, with the $d$-radial spectral density  (\ref{GWdensity}).
\end{proposition}

We term ${\cal GW}_{a,\xi,\nu,d,k}$ the $(d,k)$-hole effect Generalized Wendland model or, more briefly, the hole effect Generalized Wendland model. Using \cite[7.4.1.2]{prud} and the reflection formula for the gamma function, one can express ${\cal GW}_{a,\xi,\nu,d,k}$ in terms of sums of  Gauss hypergeometric functions ${}_2F_{1}$ instead of  generalized hypergeometric functions ${}_3F_2$:
\begin{equation}
\label{wendlandext}
\begin{split}
    &{\cal GW}_{a,\xi,\nu,d,k}(h) \\
    &= \sum_{n=0}^k \frac{(-1)^n k! (\frac{1-\nu}{2}-\xi)_n (-\frac{\nu}{2}-\xi)_n}{n! (k-n)! (1-\frac{d}{2}-n)_n (\frac{1}{2}-\xi)_n}  \left(\frac{h}{a}\right)^{2n} {{}_2F_1\left(\frac{1-\nu}{2}-\xi+n,-\xi-\frac{\nu}{2}+n;\frac{1}{2}-\xi+n;\frac{h^2}{a^2}\right)}\\
    &+ \frac{\Gamma(2\xi+1+\nu) \Gamma(-\xi-\frac{1}{2})}{(\frac{d}{2})_k \Gamma(\xi+\frac{1}{2})\Gamma(\nu) 2^{1+2\xi}} \sum_{n=0}^k \frac{(-1)^{n+k} k! (\frac{1-d}{2}-\xi-k)_{k-n} (1-\frac{\nu}{2})_n (\frac{1-\nu}{2})_n}{n! (k-n)! (\xi+\frac{3}{2})_n}\\
    &\qquad \times \left(\frac{h}{a}\right)^{2\xi+1+2n} {{}_2F_1\left(1-\frac{\nu}{2}+n,\frac{1-\nu}{2}+n;\xi+\frac{3}{2}+n;\frac{h^2}{a^2}\right)}, \quad 0 \leq h < a,
\end{split}
\end{equation}
and $0$ if $h \geq a$. This last expression allows for a numerical computation of ${\cal GW}_{a,\xi,\nu,d,k}$, because the Gauss hypergeometric function is implemented in the GNU scientific library and in the most important statistical softwares, including R, Matlab and Python. Other analytical expressions of ${\cal GW}_{a,\xi,\nu,d,k}$ in terms of Gauss hypergeometric functions, associated Legendre functions, or Meijer-$G$ functions are given in {Appendix}.\\

Compared to the standard Generalized Wendland model, the model ${\cal GW}_{a,\xi,\nu,d,k}$ has an extra discrete parameter $k$ describing increasing
levels of negative correlations when $k = 1,2....$. 
In addition, as in the hole effect Matern model, it depends on the dimension $d$.
The role of the other parameters is unchanged.

The computation of the $(d,k)$-hole effect Generalized Wendland  model using (\ref{wendlandext}) can be cumbersome to statisticians used to handle closed-form parametric correlation models. When $k=0$  and $\xi \in \mathbb{N}$, closed-form expressions of (\ref{wendlandext})
that do not depend on the Gauss hypergeometric function can be obtained (see Equation \ref{ppoo}). Similarly, the following proposition provides a closed-form solution when $k=1,2$ and $\xi \in \mathbb{N}$. \\

\begin{proposition}
\label{WendN}

If $\xi \in \mathbb{N}$ and $k=1$, then
\begin{eqnarray}
\label{genordwend}
      {\cal GW}_{a,\xi,\nu,d,1}(h) &=& \sum_{n=0}^{\xi} a_{\xi,n}(\nu) \left(1-\frac{h}{a}\right)^{n+\xi+\nu-1} \left(1+\frac{h}{a}\right)^{\xi-n-1}  \nonumber \left[1 - \frac{(2n+\nu)h}{a \,d} - \frac{(2\xi+\nu+d)h^2}{a^2 \,d} \right],
\end{eqnarray}
for $0 \leq h < a$, and $0$ otherwise,
with $a_{\xi,n}(\nu) =  \frac{2^{\nu} \Gamma(\xi)\Gamma(2\xi+\nu+1) (\nu)_n (-\xi)_n}{\Gamma(2\xi)\Gamma(\xi+\nu+1)2^{\nu+1} (\xi+\nu+1)_n \, n!}$.\\

If $\xi \in \mathbb{N}$ and $k=2$, the correlation ${\cal GW}_{a,\xi,\nu,d,2}$ is identically equal to
\begin{equation}
\label{genordwend2}
\begin{split}
      &{\cal GW}_{a,\xi,\nu,d,2}(h) =   \\
      &=  \begin{cases}
      \sum_{n=0}^{\xi} a_{\xi,n}(\nu) \left(1-\frac{h}{a}\right)_+^{n+\xi+\nu-2} \left(1+\frac{h}{a}\right)^{\xi-n-2} \\
      \Bigg[\left(1-\frac{h^2}{a^2}- \frac{h(n+\xi+\nu-1)}{a \, d} \left(1+\frac{h}{a}\right) + \frac{h(\xi-n-1)}{a \, d} \left(1-\frac{h}{a}\right)\right)
      \left(1 - \frac{(2n+\nu)h}{a \,(d+2)} - \frac{(2\xi+\nu+d+2)h^2}{a^2 \,(d+2)}\right)
       \\
      + \frac{h}{d} \left(1-\frac{h^2}{a^2}\right)\left(- \frac{(2n+\nu)}{a \,(d+2)} - \frac{2h(2\xi+\nu+d+2)}{a^2 \,(d+2)}
      \right) \Bigg], \quad 0 \leq h < a\\
      0, \quad h \geq a.\\
      \end{cases}
\end{split}
\end{equation}
\end{proposition}

As an example, let us consider the simplest version of the Generalized Wendland  model, $i.e.$ the  Askey correlation model \citep{Golubov}:
 $${\cal GW}_{a,0,\nu,d,0}(h) = {\cal GW}_{a,0,\nu}(h)= \left(1-\frac{h}{a}\right)_+^{\nu}.$$
For $k=1$ or $k=2$ and $\nu \geq \frac{d+1}{2}+k$, this correlation extends to the following models:
\begin{equation*}
    \label{Askey2}
    {\cal GW}_{a,0,\nu,d,1}(h) = \left(1-\frac{h}{a}\right)_+^{\nu-1} \left(1 - \frac{(\nu+d) \,h}{a \, d}\right)
\end{equation*}
and
\begin{equation*}
    \label{Askey3}
    {\cal GW}_{a,0,\nu,d,2}(h) = \left(1-\frac{h}{a}\right)_+^{\nu-2} \left[1-\left(2+\frac{\nu (2d+3)}{d(d+2)} \right)\frac{h}{a}+\frac{h^2}{a^2} \left(1+ \frac{\nu (2d+\nu+2)}{d(d+2)}\right) \right].
\end{equation*}

Another example of (\ref{genordwend}) is
\begin{equation*}
    \label{ciccio}
    {\cal GW}_{a,1,\nu,d,1}(h) = \left(1-\frac{h}{a}\right)_+^{\nu} \left(1 +\nu h- \frac{(\nu+1)(\nu+2+d)}{d}\frac{h^2}{a^2} \right),
\end{equation*}
which is the model in \cite[Equation 13]{gneiting2002compactly}.

An illustration of the hole effect Askey model ${\cal GW}_{a,0,\nu,d,k}$ for $k=0, 1, 2, 3$ and $d=2$ is depicted in Figure \ref{holegw} (left part). It can be seen that, when increasing $k$, the hole effect increases and the correlation is flattened at the same time. Figure \ref{holegw} (right part) depicts examples of $ {\cal GW}_{a,1,\nu,d,k}$ for $k=0, 1, 2, 3$ and $d=2$. One can appreciate the different levels of differentiability at the origin between the  examples of the left and right parts.

\begin{figure}[h!]
\begin{center}
\begin{tabular}{cc}
\includegraphics[width=6.2cm,height=6.2cm]{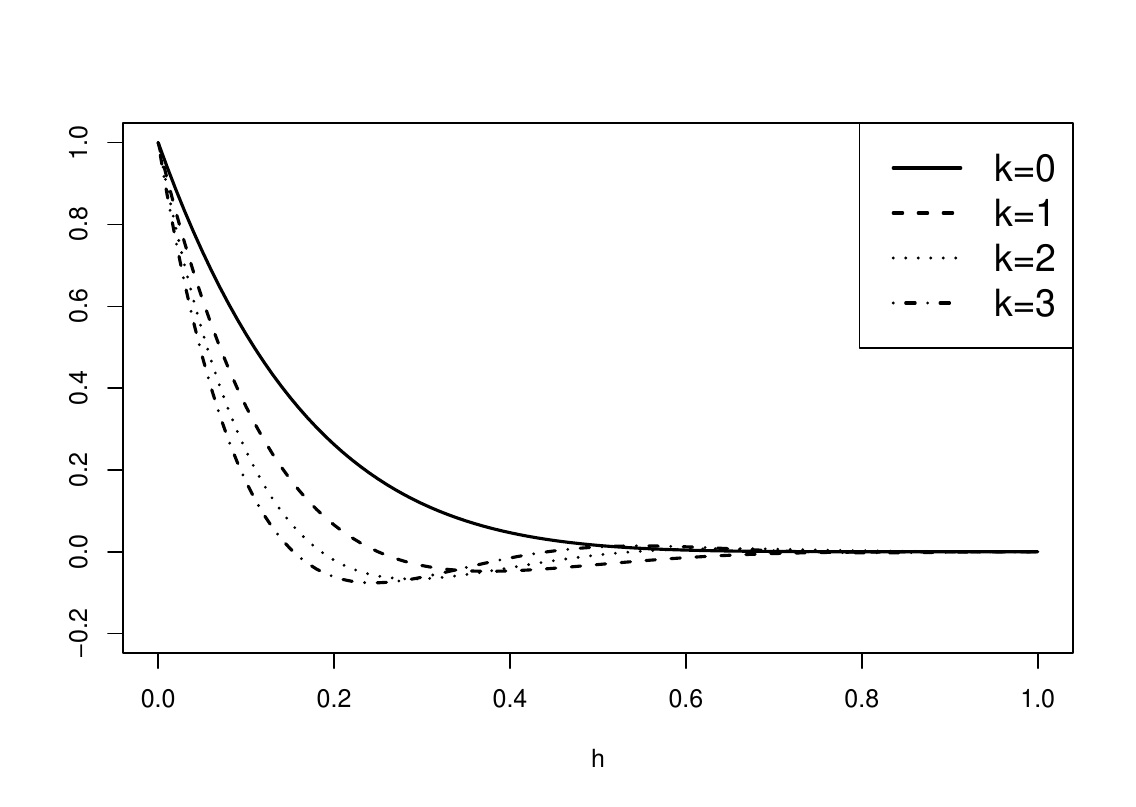}&\includegraphics[width=6.2cm,height=6.2cm]{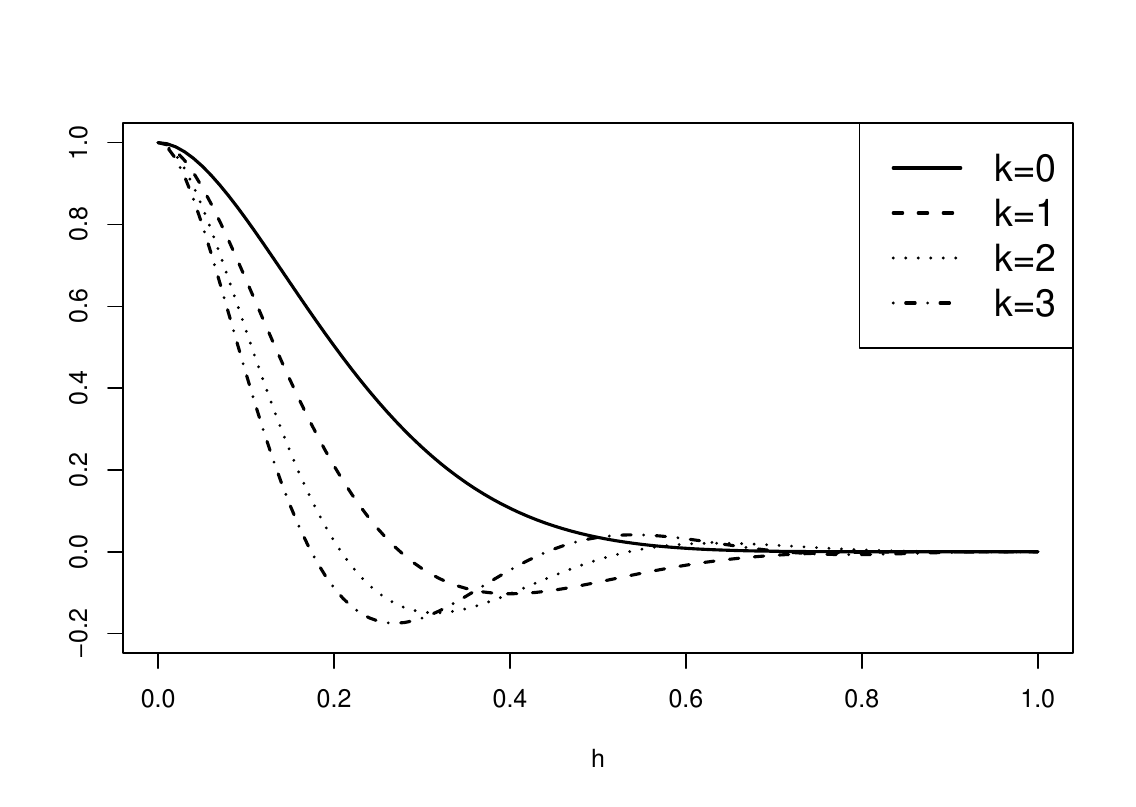}
\end{tabular}
\end{center}
\caption{Examples of the $(d,k)$-hole  Generalized Wendland model when $d=2$. Left: ${\cal GW}_{1,0,6,2,k}(\cdot)$ for $k=0,1,2,3$ from top to bottom. Right: ${\cal GW}_{1,1,6,2,k}(\cdot)$ for $k=0,1,2,3$ from top to bottom.}
 \label{holegw}
\end{figure}

Finally, Figure \ref{realizations} depicts realizations of zero-mean and unit-variance Gaussian random fields with correlations ${\cal GW}_{0.4,0,6,2,k}$ and ${\cal GW}_{0.4,1,6,2,k}$ for $k=0,1,2$. The realizations were constructed via the Cholesky decomposition of the covariance matrix.

\begin{figure}[h!]
\begin{center}
\begin{tabular}{ccc}
\includegraphics[width=4.2cm,height=5.2cm]{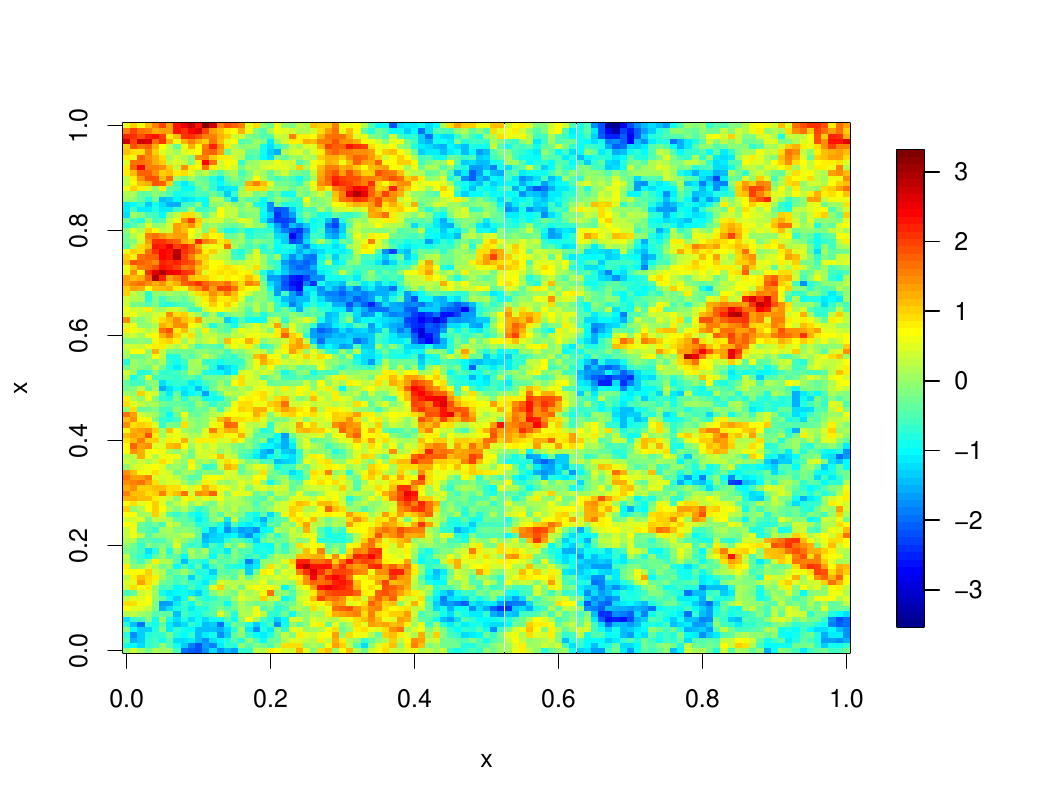}&\includegraphics[width=4.2cm,height=5.2cm]{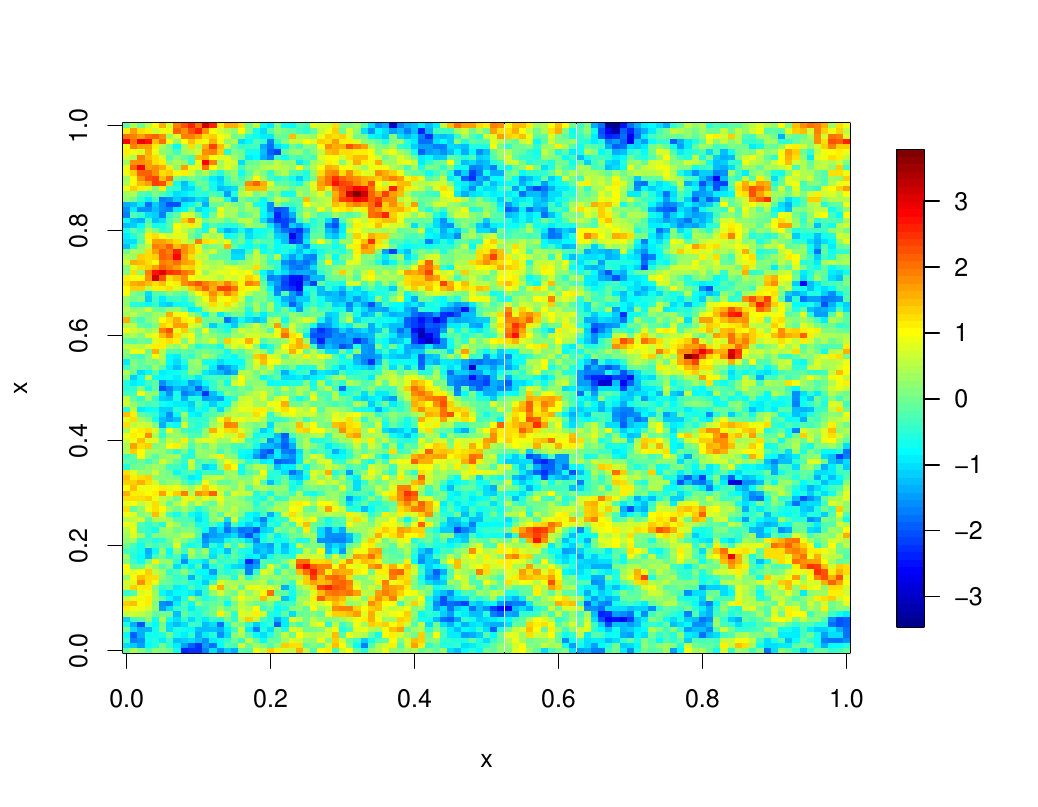}&\includegraphics[width=4.2cm,height=5.2cm]{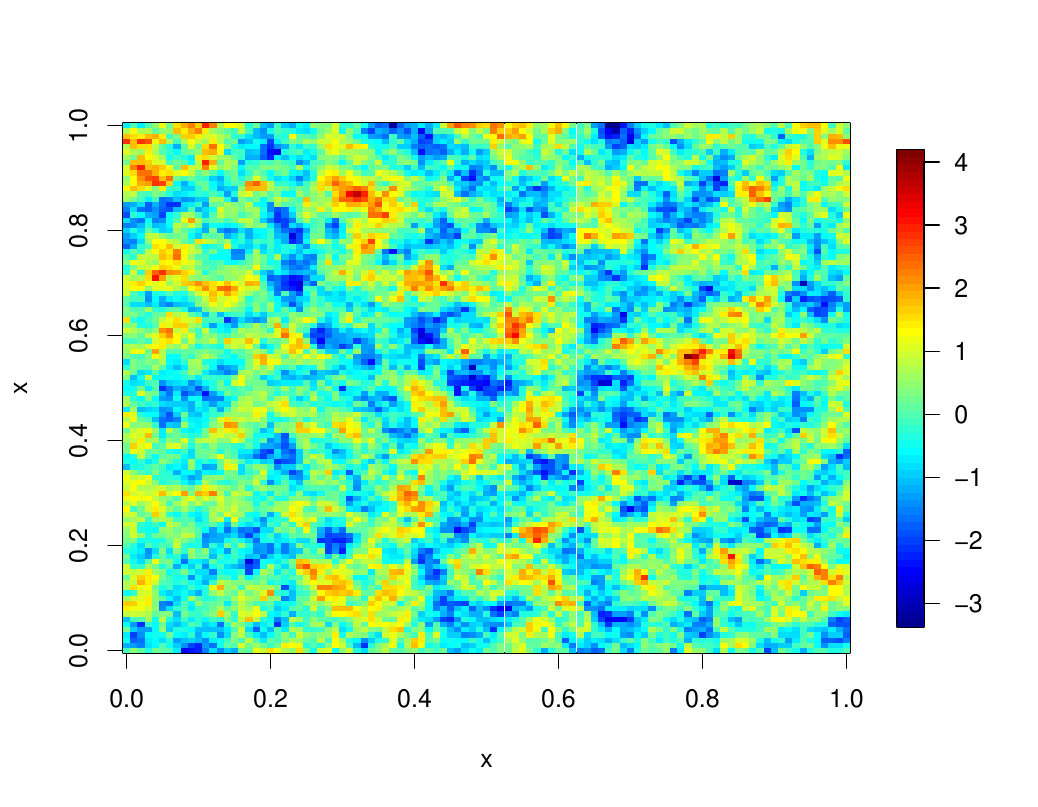}\\
\includegraphics[width=4.2cm,height=5.2cm]{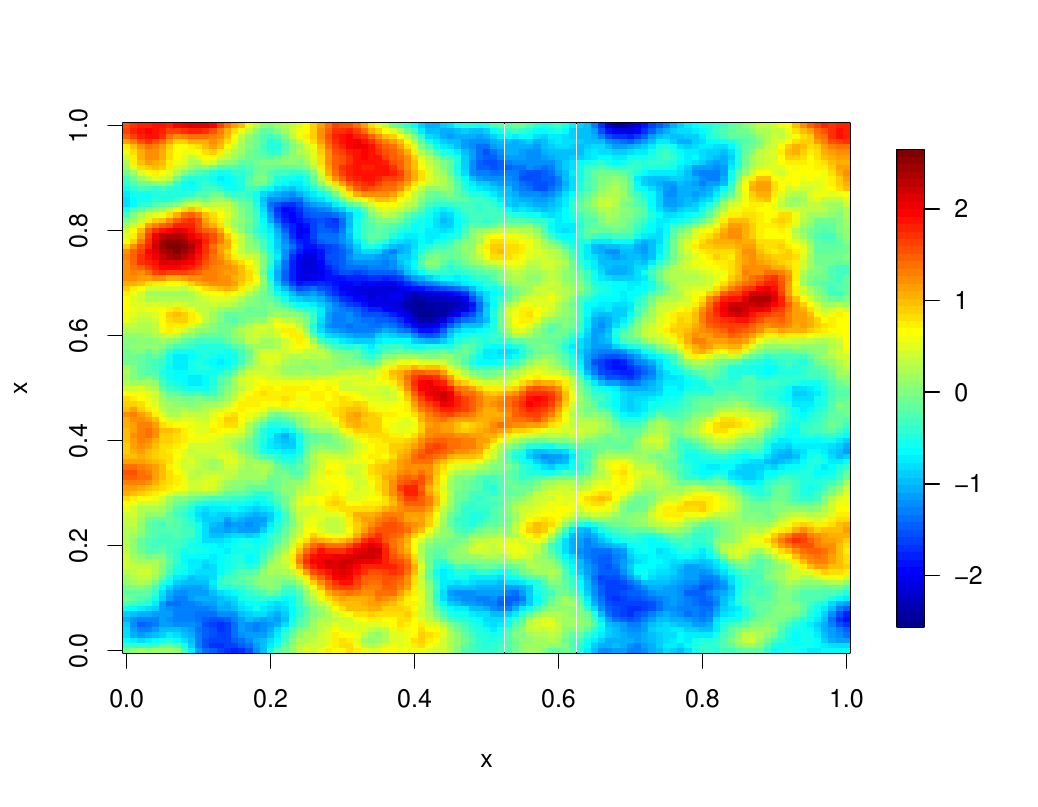}&\includegraphics[width=4.2cm,height=5.2cm]{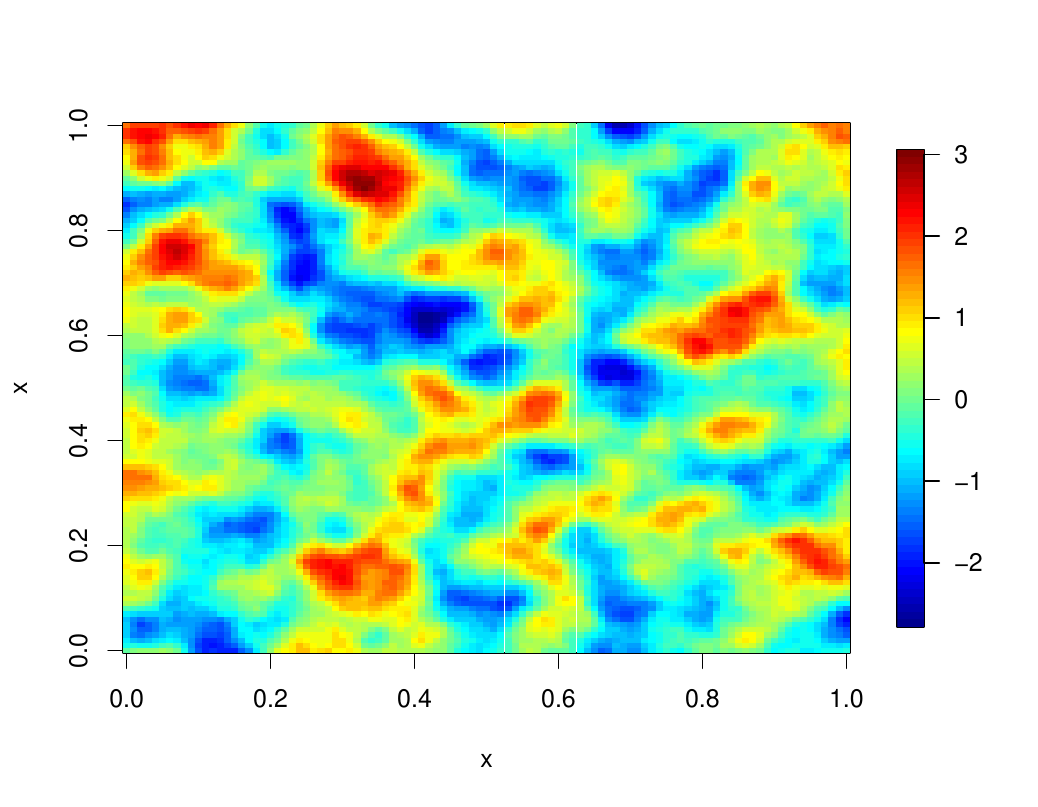}&\includegraphics[width=4.2cm,height=5.2cm]{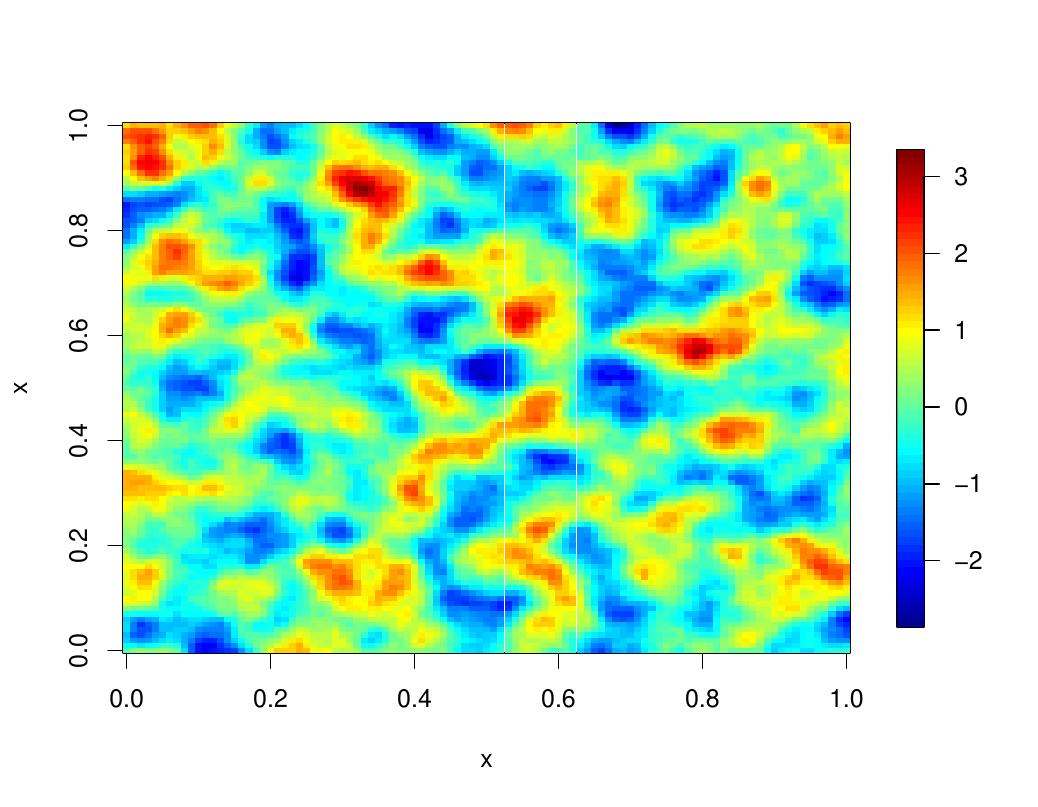}\\
\end{tabular}
\end{center}
\caption{
Top: three Gaussian random field realizations with ${\cal GW}_{0.4,0,6,2,k}$ correlation model, for $k=0,1,2$ (from left to right). Bottom: three Gaussian random field realizations with ${\cal GW}_{0.4,1,6,2,k}$ correlation model, for $k=0,1,2$ (from left to right).}
\label{realizations}
\end{figure}

\subsection{A Bridge Between Compactly and Globally Supported Models with Hole Effects} \label{seccc5}

The following shows that our compactly and globally supported models that parameterize hole effects can be put under the same umbrella. The key is an elegant convergence argument proving that a reparameterized version of the $(d,k)$-hole effect Generalized Wendland model converges to the $(d,k)$-hole effect Mat{\'e}rn model for every fixed $k$. The result is formally stated below.

\begin{proposition}
\label{wend2mat}
    Let $a$, $\xi$, $\nu > 0$ and $k \in \mathbb{N}$. As $\nu$ tends to $+\infty$, ${\cal GW}_{\nu a,\xi-{1/2},\nu,d,k}$ and $\widehat{{\cal GW}}_{\nu a,\xi-{1/2},\nu,d,k}$ uniformly converge to ${\cal M}_{a,\xi,d,k}$ and $\widehat{{\cal M}}_{a,\xi,d,k}$ on $[0,+\infty)$:
    \begin{equation}
    \label{ppoi4}
    \lim_{\nu\to\infty}  {\cal GW}_{\nu a,\xi-{1/2},\nu,d,k}(h)={\cal M}_{a,\xi,d,k}(h), \quad h \geq 0.
    \end{equation}
    \begin{equation*}
    \label{ppoi5}
    \lim_{\nu\to\infty}  \widehat{{\cal GW}}_{\nu a,\xi-{1/2},\nu,d,k}(u)=\widehat{{\cal M}}_{a,\xi,d,k}(u), \quad u \geq 0.
    \end{equation*}
\end{proposition}

For $k=0$, the claim of this proposition is similar to the result (\ref{ppoi}) of \cite{bevilacqua2022unifying}. However, the proposed parameterization of the compact support has here a clearer interpretation, as it does not involve the smoothness parameter as in (\ref{ppoi}).
A consequence of Proposition $\ref{wend2mat}$ is that, for a given smoothness parameter $\xi$ and scale parameter $a$, the parameter $\nu$ fixes the sparseness of the associated correlation matrices and allows {switching} from the world of flexible compactly supported correlation models with hole effects to the world of flexible globally supported correlation models with hole effects.

As an illustration, Figure \ref{convergence} depicts ${\cal GW}_{\nu a,\xi-{1/2},\nu,d,k}$ when $k=0,1,2$ and $\nu=10,50$ or $\nu\to\infty$, the latter being ${\cal M}_{a,\xi,d,k}$, for some specific values of the parameters $(a,d,\xi)$. When $k$ and $\xi$ increase at the same time, both ${\cal GW}_{\nu a,\xi-{1/2}, \nu, d,k}$ and ${\cal M}_{a,\xi ,d,k}$ behave as differentiable (at the origin) oscillating correlation functions.

\begin{figure}[h!]
\begin{center}
\begin{tabular}{ccc}
\includegraphics[width=4.5cm,height=5.2cm]{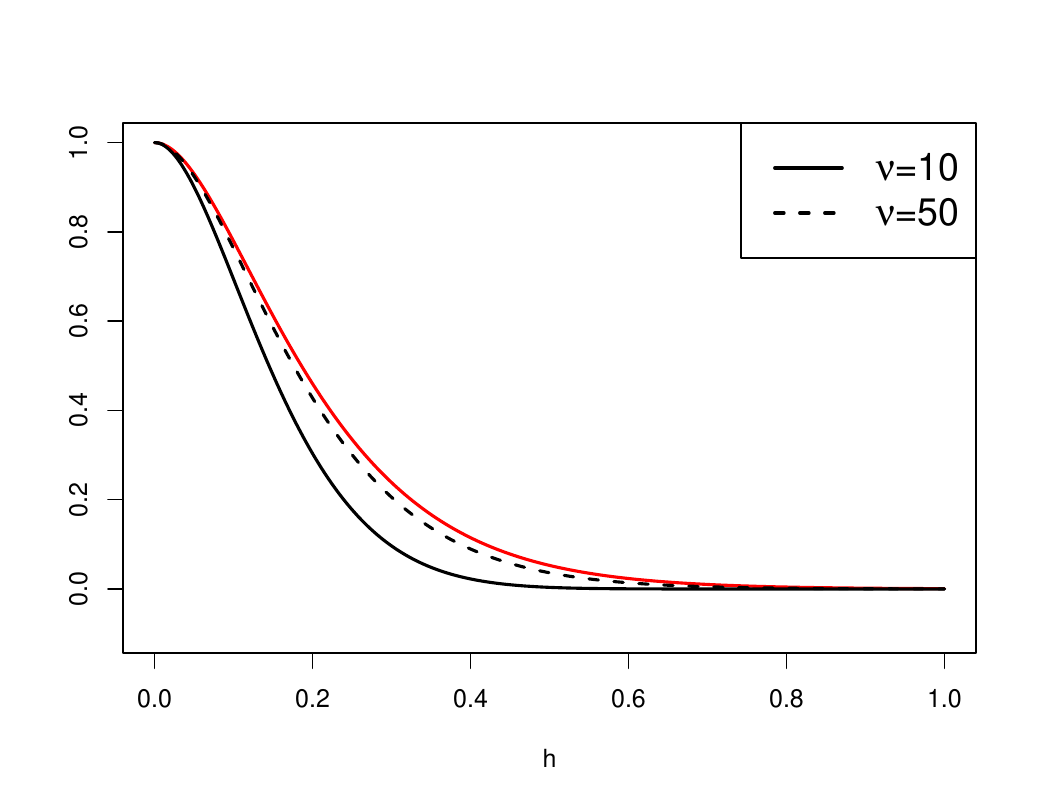}&\includegraphics[width=4.5cm,height=5.2cm]{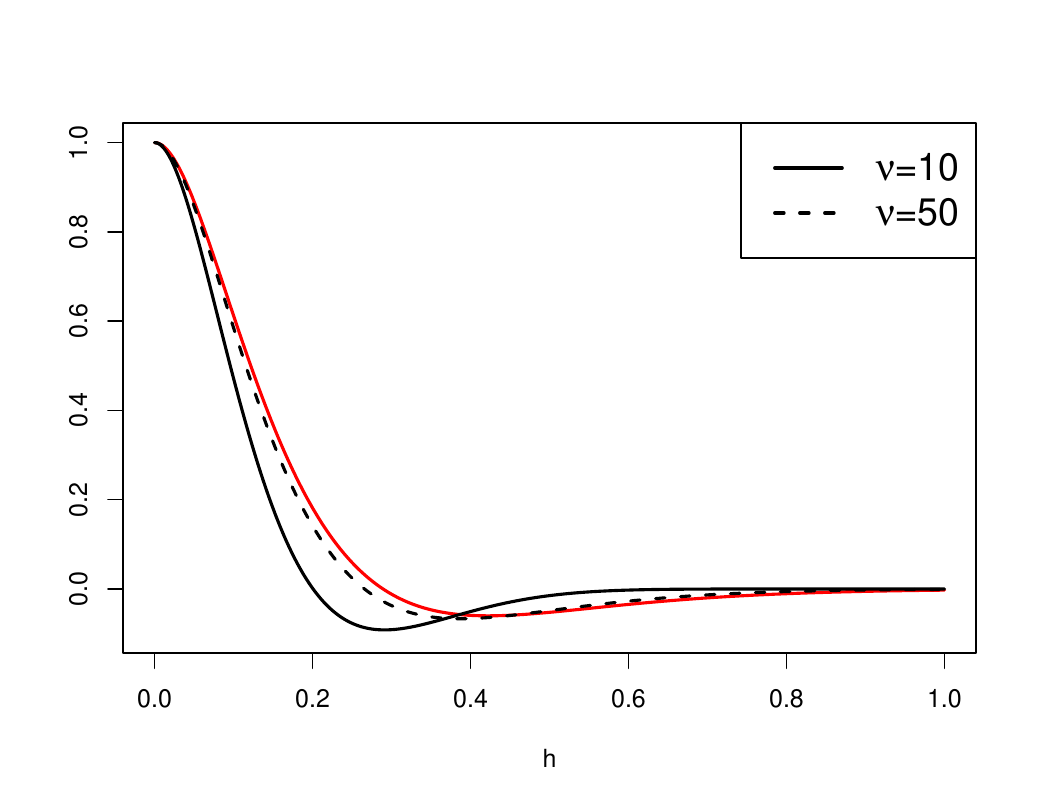}&\includegraphics[width=4.5cm,height=5.2cm]{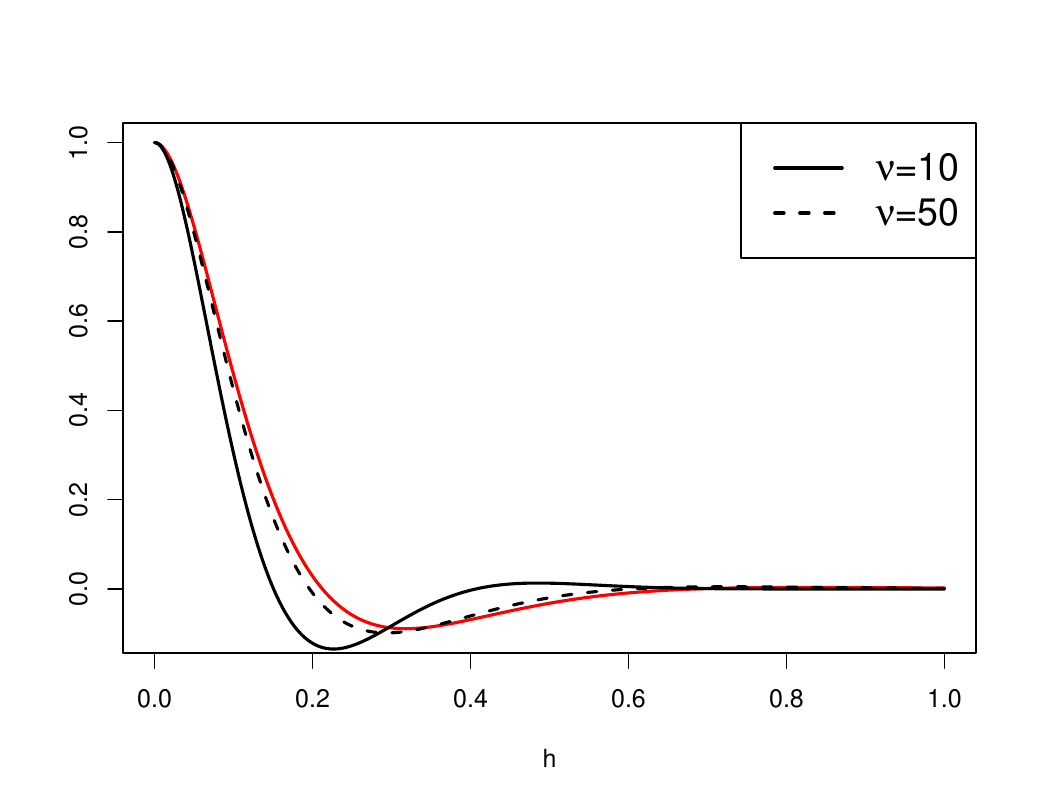}\\
\end{tabular}
\end{center}
\caption{
Left: ${\cal GW}_{\nu a,\xi-{1/2}, \nu,d,k}$ for $\nu=10,50$ and ${\cal M}_{a,\xi,d,k}$ (red line) when $\xi=1.75$, $a=0.1$, $d=2$ and $k=0$. Center and right: the same correlations models with $k=1,2$.}
 \label{convergence}
\end{figure}

\section{Estimation of the hole effect Generalized Wendland model: A simulation study}
\label{sec:estim}

We analyze the performance of the maximum likelihood estimation of the parameters of the $(d,k)$-hole effect Generalized Wendland model using the parameterization proposed in Section \ref{seccc5}, that is we consider the covariance model $\sigma^2 {\cal GW}_{\nu a,\xi-{1/2}, \nu,d,k}$  where $\sigma^2>0$ is a variance parameter. Hereafter, we focus on the planar case, $i.e.$ we set $d=2$, and we assume a scenario where the {continuous} parameters can be consistently estimated, $i.e.$ an increasing domain scenario.

Being an integer, the hole effect parameter $k$ is considered fixed and the goal of the simulation study is to investigate if the hole effect affects the maximum likelihood estimation of the parameters $a, \sigma^2, \xi$ that can vary continuously. {In practice, different plausible values of $k$ can be explored, and the estimations of $a, \sigma^2, \xi$ can be compared based on criteria such as the log-likelihood, Akaike information criterion (AIC), or predictive performance, as will be shown in the real data analysis in Section \ref{sec:real}.}

{Also, for the $\sigma^2{\cal GW}_{\nu a,\xi - {1/2},\nu,d,k}$ covariance model, $\nu$ determines the compact support of the covariance function under our proposed parameterization. In particular, for small values of $\nu$, the covariance matrix becomes highly sparse, whereas as $\nu \to \infty$, the model approaches the hole-effect Matérn covariance with a fully dense covariance matrix. As a consequence, $\nu$ can either be fixed by the user—when sparse matrices are desired for computational efficiency—or estimated from the data. In this section, we adopt the former approach, whereas in the data application in Section \ref{sec:real}, we will follow the latter.}

Maximum likelihood estimation partially takes advantage of the computational benefits of the $\sigma^2{\cal GW}_{\nu a,\xi-{1/2},\nu,d,k}$ model because the compact support $\nu a$ depends on $a$ and $\nu$. Even when considering a fixed $\nu$, the covariance matrix can be highly
or slightly sparse, depending on the value of $a$ in the optimization process. Alternative methods of estimation with a good balance between statistical efficiency and computational complexity include, among others, composite likelihood \citep{chaa,compblock}, multi-resolution approximation \citep{nychka2}, or methods based on directed acyclic graphs using Vecchia’s approximations \citep{ka2021,spabion2024}.

Using the Cholesky decomposition method, we simulate $500$ realizations of a zero-mean Gaussian random field with covariance $\sigma^2{\cal GW}_{\nu a,\xi-{1/2},\nu,d,k}$ observed at 1000 locations uniformly distributed in a unit square.
We consider different scenarios with increasing smoothness parameters $\xi=0.5,1.5,2.5$ and increasing levels of negative correlations  $k=0,1,2$. In addition  we set   $a=0.1$, $\sigma^2=1$ and $\nu=6$. For each realization, we estimate the parameters $(a, \sigma^2,\xi)$ with maximum likelihood.

Table \ref{tabl14} reports the bias and 
{root}
mean squared error (RMSE) associated with the estimation of $(a,\sigma^2,\xi)$. Overall the bias is approximatively zero for each combination of the parameters and in general the patterns observed for the case $k=0$ are also observed when the covariance model has a hole effect, $i.e.$ when $k=1,2$. For instance, the RMSE of the scale parameter  $a$ decreases when increasing $\xi$ for each $\nu$  and for each $k=0,1,2,$ and the RMSE of the smoothness parameter $\xi$  decreases when increasing $\xi$ for each $\nu$  and for each $k=0,1,2$.

In summary, our numerical experiments show that the hole effect does not affect the maximum likelihood  estimation of the $(d,k)$-hole effect reparameterized Generalized Wendland covariance model. As an illustration, Figure \ref{bxplots} depicts the boxplots of the maximum likelihood estimates for the parameters $(a,\sigma^2,\xi)$ (from left to right) when estimating the covariance model $\sigma^2{\cal GW}_{20a,\xi-{1/2}, 20, 2,k}$ with increasing levels of negative correlations $k=0,1,2$ when $a=0.1$, $\sigma^2=1$ and $\xi=1.5$. {It can be appreciated that the boxplots are nearly identical across different values of $k$, indicating that the estimates of $a$, $\sigma^2$ and $\xi$ are robust to the choice of the hole effect parameter $k$.}

\begin{table}[]
  \caption{Bias and {root} mean squared error (in parentheses) of the maximum likelihood estimates of the parameters $(a,\sigma^2,\xi)$
  of the correlation model $\sigma^2{\cal GW}_{\nu a,\xi-{1/2},\nu,2,k}$ for increasing levels of negative correlations $k=0,1,2$, increasing smoothness parameter $\xi=0.5,1.5,2.5$,
  and increasing levels of $\nu=5,20,\infty$.}
  \label{tabl14}

  \begin{center}
  \scalebox{0.65}{

\begin{tabular}{cc|ccc|ccc|ccc|}

\cline{3-11}&      & \multicolumn{3}{c|}{$\nu=5$}  & \multicolumn{3}{c|}{$\nu=20$} & \multicolumn{3}{c|}{$\nu=\infty$} \\ \hline
\multicolumn{1}{|c|}{}                            &                             & \multicolumn{1}{c|}{$\xi=0.5$}   & \multicolumn{1}{c|}{$\xi=1.5$}    & $\xi=2.5$   & \multicolumn{1}{c|}{$\xi=0.5$}   & \multicolumn{1}{c|}{$\xi=1.5$} & $\xi=2.5$ & \multicolumn{1}{c|}{$\xi=0.5$}   & \multicolumn{1}{c|}{$\xi=1.5$} & $\xi=2.5$   \\ \hline
\multicolumn{1}{|c|}{\multirow{6}{*}{$k=0$}} & \multirow{2}{*}{$a$}   & \multicolumn{1}{c|}{$0.0003$}   & \multicolumn{1}{c|}{$-0.0004$}   & $-0.0003$  & \multicolumn{1}{c|}{$0.0009$}   & \multicolumn{1}{c|}{-0.0004}          &    \multicolumn{1}{c|}{$-0.0049$}        & \multicolumn{1}{c|}{$0.0011$}   & \multicolumn{1}{c|}{$-0.0003$}          &     \multicolumn{1}{c|}{-0.0005}         \\
\multicolumn{1}{|c|}{}                            &                            & \multicolumn{1}{c|}{$(0.0058)$} & \multicolumn{1}{c|}{$(0.0028)$} & $(0.0017)$ & \multicolumn{1}{c|}{$(0.0069)$}   & \multicolumn{1}{c|}{(0.0037)}          &         \multicolumn{1}{c|}{$(0.0024)$}       & \multicolumn{1}{c|}{$(0.0074)$} & \multicolumn{1}{c|}{$(0.0042)$}          &     \multicolumn{1}{c|}{$(0.0028)$}          \\ \cline{2-11}
\multicolumn{1}{|c|}{}                            & \multirow{2}{*}{$\sigma^2$} & \multicolumn{1}{c|}{$0.004$}   & \multicolumn{1}{c|}{$0.006$}    & $0.005$   & \multicolumn{1}{c|}{$0.007$}   & \multicolumn{1}{c|}{$0.013$}          &     \multicolumn{1}{c|}{$0.011$}       & \multicolumn{1}{c|}{$0.008$}   & \multicolumn{1}{c|}{0.017}          &    \multicolumn{1}{c|}{$0.018$}           \\
\multicolumn{1}{|c|}{}                            &                             & \multicolumn{1}{c|}{$(0.015)$} & \multicolumn{1}{c|}{$(0.021)$}  & $(0.021)$ & \multicolumn{1}{c|}{$(0.016)$} & \multicolumn{1}{c|}{(0.026)}          &      \multicolumn{1}{c|}{$(0.029)$}       & \multicolumn{1}{c|}{$(0.016)$} & \multicolumn{1}{c|}{(0.028)}          &      \multicolumn{1}{c|}{$(0.032)$}            \\ \cline{2-11}
\multicolumn{1}{|c|}{}                            & \multirow{2}{*}{$\xi$}      & \multicolumn{1}{c|}{$0.0028$}   & \multicolumn{1}{c|}{$0.0036$}    & $0.0035$   & \multicolumn{1}{c|}{$0.0039$}   & \multicolumn{1}{c|}{0.0050}          &        \multicolumn{1}{c|}{$0.0051$}    & \multicolumn{1}{c|}{$0.0042$}   & \multicolumn{1}{c|}{$0.0057$}          &  \multicolumn{1}{c|}{$0.0057$}            \\
\multicolumn{1}{|c|}{}                            &                             & \multicolumn{1}{c|}{$(0.0030)$} & \multicolumn{1}{c|}{$(0.0043)$}  & $(0.0049)$ & \multicolumn{1}{c|}{$(0.0032)$} & \multicolumn{1}{c|}{(0.0047)}          &     \multicolumn{1}{c|}{$(0.0054)$}        & \multicolumn{1}{c|}{$(0.0033)$}   & \multicolumn{1}{c|}{(0.0050)}          &     \multicolumn{1}{c|}{$(0.0057)$}          \\ \hline
\multicolumn{1}{|c|}{\multirow{6}{*}{$k=1$}} & \multirow{2}{*}{$a$}   & \multicolumn{1}{c|}{$-0.0008$}  & \multicolumn{1}{c|}{$-0.0003$}   & $-0.0002$  & \multicolumn{1}{c|}{$-0.0005$}  & \multicolumn{1}{c|}{-0.0005}          &  \multicolumn{1}{c|}{-0.0004}           & \multicolumn{1}{c|}{$-0.0003$}  & \multicolumn{1}{c|}{$-0.0001$}          & $-0.0003$  \\
\multicolumn{1}{|c|}{}                            &                             & \multicolumn{1}{c|}{$(0.0022)$} & \multicolumn{1}{c|}{$(0.0017)$}  & $(0.0010)$ & \multicolumn{1}{c|}{$(0.0041)$} & \multicolumn{1}{c|}{(0.0028)}          &   \multicolumn{1}{c|}{(0.0020)}          & \multicolumn{1}{c|}{$(0.0046)$} & \multicolumn{1}{c|}{$(0.0032)$}          & $(0.0022)$ \\ \cline{2-11}
\multicolumn{1}{|c|}{}                            & \multirow{2}{*}{$\sigma^2$} & \multicolumn{1}{c|}{$-0.002$}  & \multicolumn{1}{c|}{$0.002$}    & $0.005$   & \multicolumn{1}{c|}{$-0.000$}  & \multicolumn{1}{c|}{-0.003}          &   \multicolumn{1}{c|}{-0.006}          & \multicolumn{1}{c|}{$0.001$}   & \multicolumn{1}{c|}{$0.001$}          & $0.014$   \\
\multicolumn{1}{|c|}{}                            &                             & \multicolumn{1}{c|}{$(0.009)$} & \multicolumn{1}{c|}{$(0.014)$}  & $(0.010)$ & \multicolumn{1}{c|}{$(0.010)$} & \multicolumn{1}{c|}{(0.018)}          &   \multicolumn{1}{c|}{(0.021)}          & \multicolumn{1}{c|}{$(0.010)$} & \multicolumn{1}{c|}{$(0.020)$}          & $(0.025)$ \\ \cline{2-11}
\multicolumn{1}{|c|}{}                            & \multirow{2}{*}{$\xi$}      & \multicolumn{1}{c|}{$0.0027$}   & \multicolumn{1}{c|}{$0.0030$}    & $0.0042$   & \multicolumn{1}{c|}{$0.0035$}   & \multicolumn{1}{c|}{0.0048}          &     \multicolumn{1}{c|}{0.0050}        & \multicolumn{1}{c|}{$0.0036$}   & \multicolumn{1}{c|}{$0.0007$}          & $0.0052$   \\
\multicolumn{1}{|c|}{}                            &                             & \multicolumn{1}{c|}{$(0.0030)$} & \multicolumn{1}{c|}{$(0.0043)$}  & $(0.0046)$ & \multicolumn{1}{c|}{$(0.0032)$} & \multicolumn{1}{c|}{(0.0048)}          &    \multicolumn{1}{c|}{(0.0054)}        & \multicolumn{1}{c|}{$(0.0033)$} & \multicolumn{1}{c|}{$(0.0050)$}          & $(0.0057)$ \\ \hline
\multicolumn{1}{|c|}{\multirow{6}{*}{$k=2$}} & \multirow{2}{*}{$a$}   & \multicolumn{1}{c|}{$-0.0007$}  & \multicolumn{1}{c|}{-0.0002}             &  \multicolumn{1}{c|}{-0.0000}             & \multicolumn{1}{c|}{$-0.0009$}  & \multicolumn{1}{c|}{-0.0004}          &     \multicolumn{1}{c|}{-0.0004}        & \multicolumn{1}{c|}{$-0.0006$}            & \multicolumn{1}{c|}{$0.0001$}          &     \multicolumn{1}{c|}{-0.0000}          \\
\multicolumn{1}{|c|}{}                            &                             & \multicolumn{1}{c|}{$(0.0022)$} & \multicolumn{1}{c|}{$(0.0014)$}             &   \multicolumn{1}{c|}{$(0.0010)$}           & \multicolumn{1}{c|}{$(0.0032)$} & \multicolumn{1}{|c|}{(0.0022)}                            &\multicolumn{1}{|c|}{(0.0017)}           & \multicolumn{1}{c|}{$(0.0036)$}            & \multicolumn{1}{c|}{$(0.0026)$}          &     \multicolumn{1}{c|}{$(0.0020)$}         \\ \cline{2-11}
\multicolumn{1}{|c|}{}                            & \multirow{2}{*}{$\sigma^2$} & \multicolumn{1}{c|}{$-0.0006$}  & \multicolumn{1}{c|}{0.0029}             &  \multicolumn{1}{c|}{0.0097}              & \multicolumn{1}{c|}{$0.0010$}   & \multicolumn{1}{c|}{0.0015}          &  \multicolumn{1}{c|}{0.0012}             & \multicolumn{1}{c|}{$-0.0001$}            & \multicolumn{1}{c|}{$-0.0003$}          &   \multicolumn{1}{c|}{$-0.0000$}            \\
\multicolumn{1}{|c|}{}                            &                             & \multicolumn{1}{c|}{$(0.007)$} & \multicolumn{1}{c|}{$(0.011)$}             & \multicolumn{1}{c|}{$(0.011)$}                & \multicolumn{1}{c|}{$(0.008)$} & \multicolumn{1}{|c|}{(0.014)}                            &    \multicolumn{1}{|c|}{(0.016)}        & \multicolumn{1}{c|}{(0.009)}            & \multicolumn{1}{c|}{$(0.016)$}          & \multicolumn{1}{c|}{$(0.021)$}               \\ \cline{2-11}
\multicolumn{1}{|c|}{}                            & \multirow{2}{*}{$\xi$}      & \multicolumn{1}{c|}{$0.0026$}   & \multicolumn{1}{c|}{0.0029}             &    \multicolumn{1}{c|}{0.0036}            & \multicolumn{1}{c|}{$0.0038$}   & \multicolumn{1}{c|}{$0.0043$}          & \multicolumn{1}{c|}{$0.0049$}           & \multicolumn{1}{c|}{$0.0035$}            & \multicolumn{1}{c|}{$-0.0008$}          &   \multicolumn{1}{c|}{$-0.0002$}           \\
\multicolumn{1}{|c|}{}                            &                             & \multicolumn{1}{c|}{$(0.0030)$} & \multicolumn{1}{c|}{$(0.0043)$}             &    \multicolumn{1}{c|}{$(0.0047)$}            & \multicolumn{1}{c|}{$(0.0033)$} & \multicolumn{1}{c|}{(0.0048)}          &     \multicolumn{1}{c|}{(0.0054)}       & \multicolumn{1}{c|}{$(0.0033)$}            & \multicolumn{1}{c|}{$(0.0051)$}          &       \multicolumn{1}{c|}{$(0.0059)$}       \\ \hline
\end{tabular}}
  \end{center}
\end{table}

\begin{figure}[h!]
\begin{center}
\begin{tabular}{ccc}
\includegraphics[width=4.5cm,height=5.2cm]{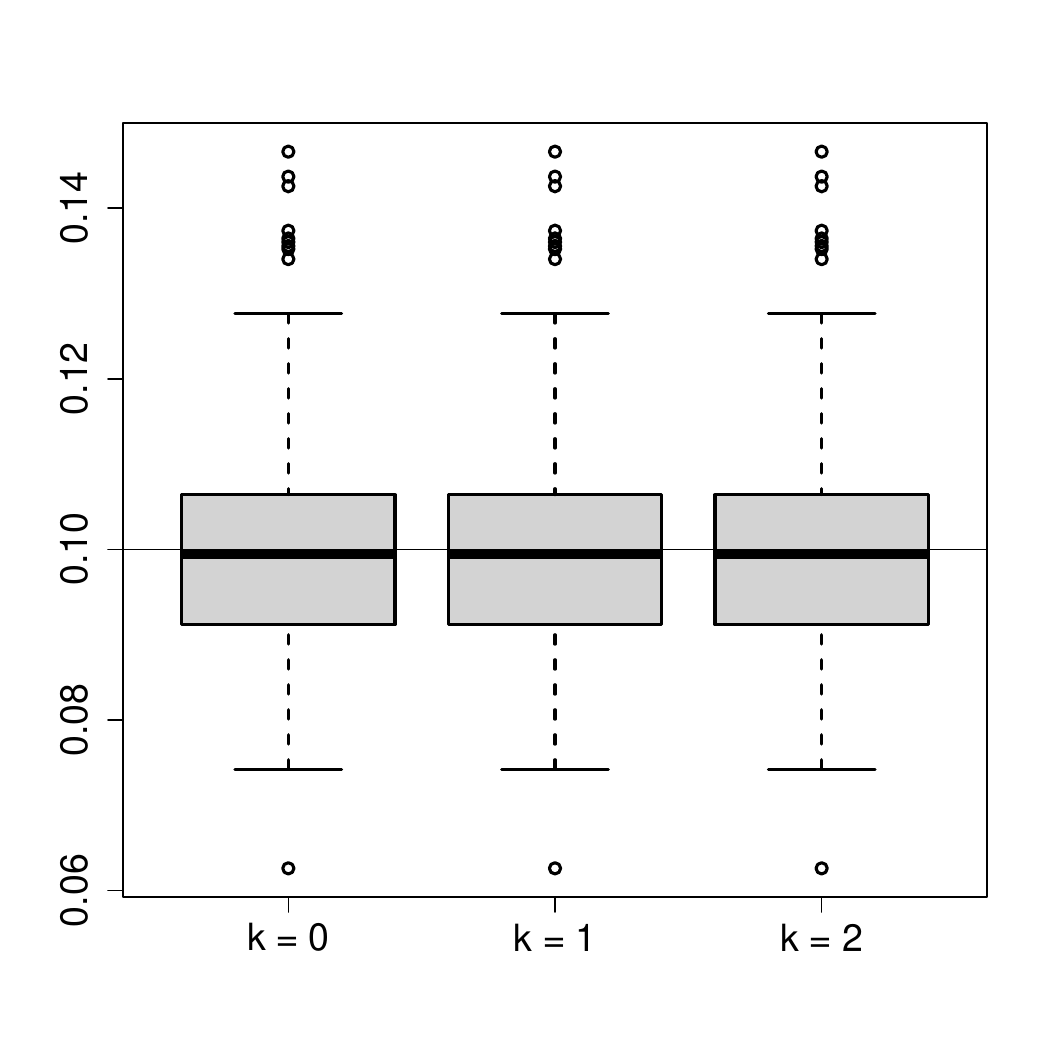}&\includegraphics[width=4.5cm,height=5.2cm]{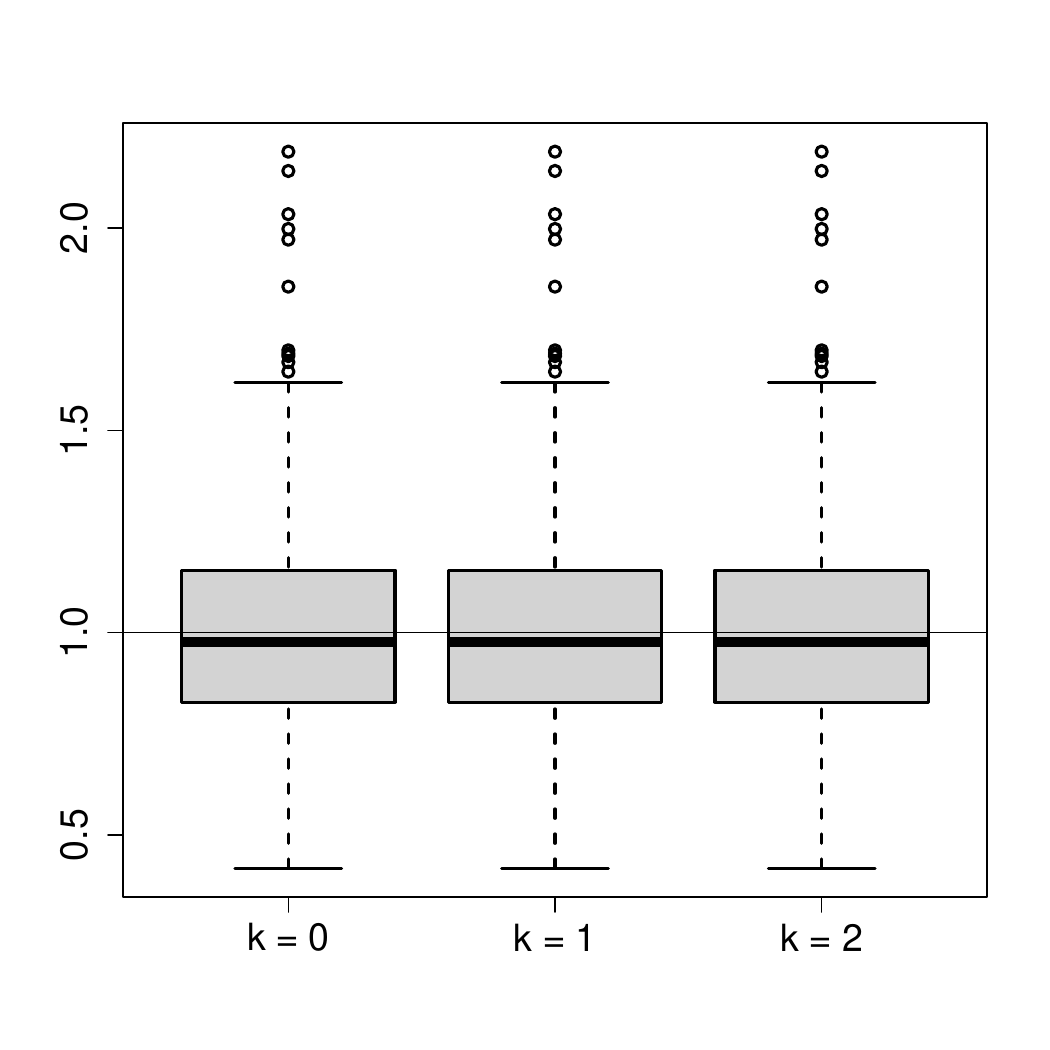}&\includegraphics[width=4.5cm,height=5.2cm]{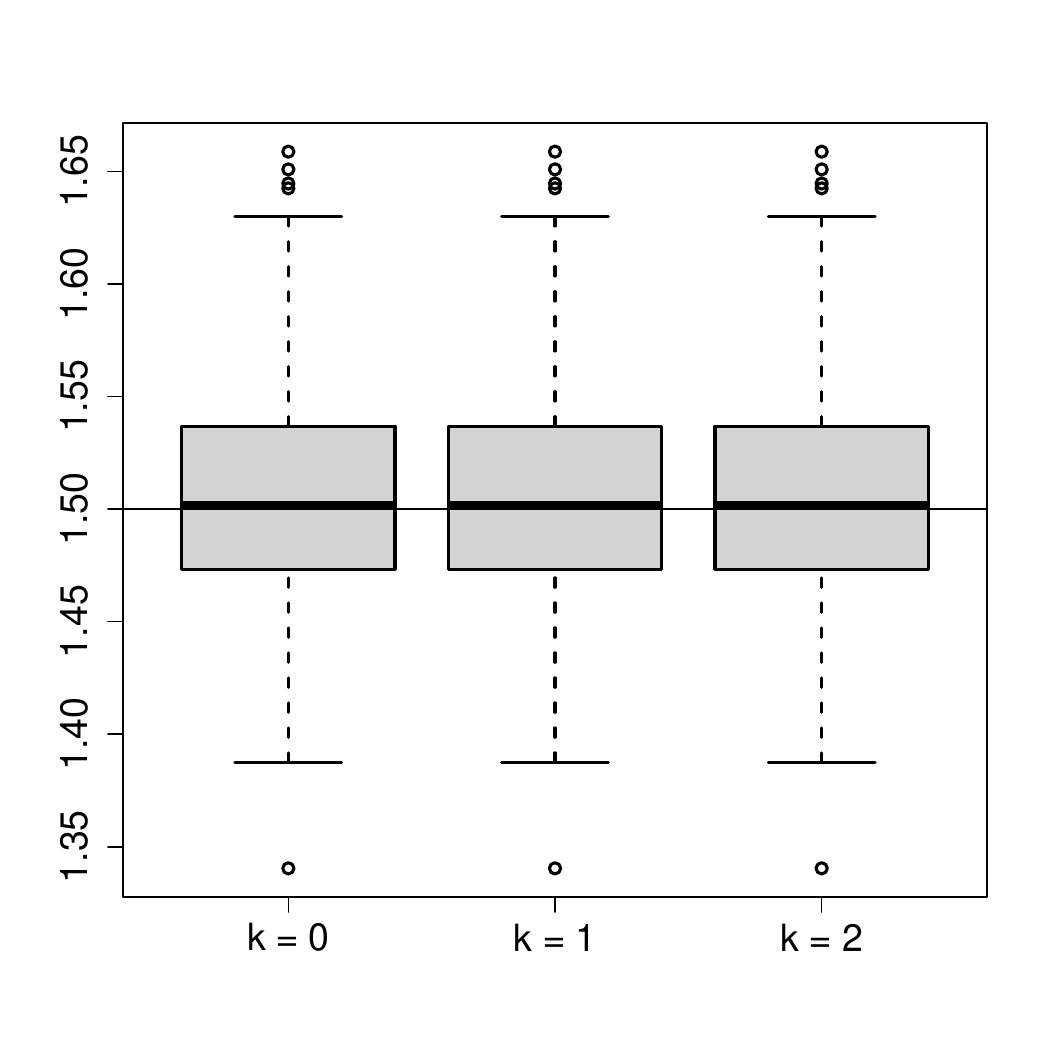}\\
\end{tabular}
\end{center}
\caption{Boxplots of the maximum likelihood  estimates of the scale dependence  parameter $a$ (left), variance parameter $\sigma^2$  (center) and smoothness parameter $\xi$ (right)
when estimating the covariance model  $\sigma^2{\cal GW}_{20a,\xi-{1/2},20,2,k}$ for $k=0,1,2$ when $a=0.1$, $\sigma^2=1$ and $\xi=1.5$.}
 \label{bxplots}
\end{figure}

\section{Real Data Illustration}
\label{sec:real}

We now consider a pedological dataset consisting of measurements of soil pH between 60 and 100 cm depth taken all over Madagascar island, {where each measurement is indexed by its easting and northing relative coordinates expressed in kilometers (i.e., $d=2$ throughout this section)}. In precision agriculture, soil pH is an important variable to assess the availability of micronutrients like nitrogen, phosphorus, and potassium to plants, and to decide whether it is necessary to neutralize soil acidity or alkalinity so as to ensure the most productive agricultural soils. The data can be downloaded using the R package  \emph{geodata}  \citep{geodata} and are documented in \cite{bokk}.

To reduce the computational burden of maximum likelihood estimation, we select a random  sample of $3,000$ locations from the original dataset. Then, following \cite{Li:Zhang:2011}, we detrend the data using splines to remove the {large-scale patterns} along the first and second coordinates.
\begin{figure}[h!]
\begin{center}
\begin{tabular}{ccc}
\includegraphics[width=5cm,height=5.8cm]{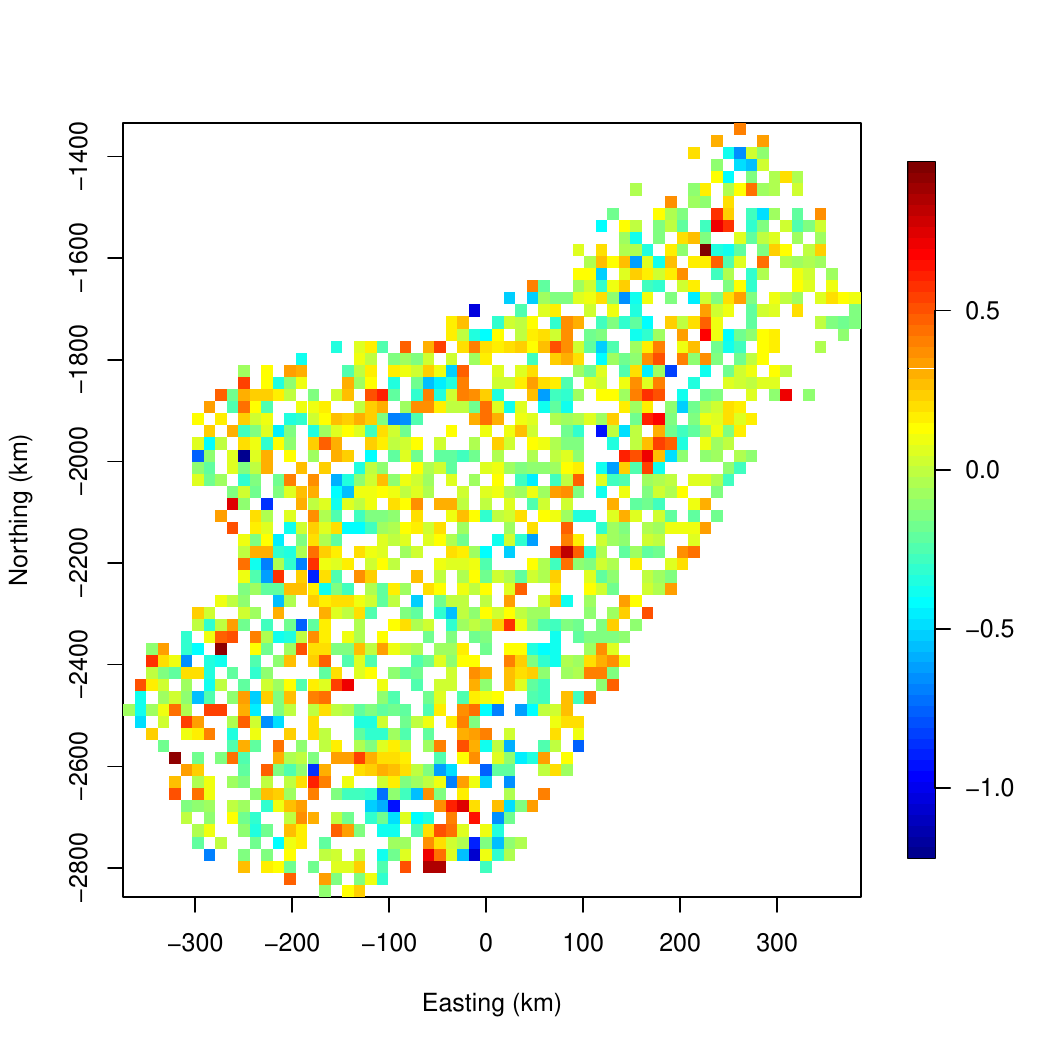}&\includegraphics[width=5cm,height=5.8cm]{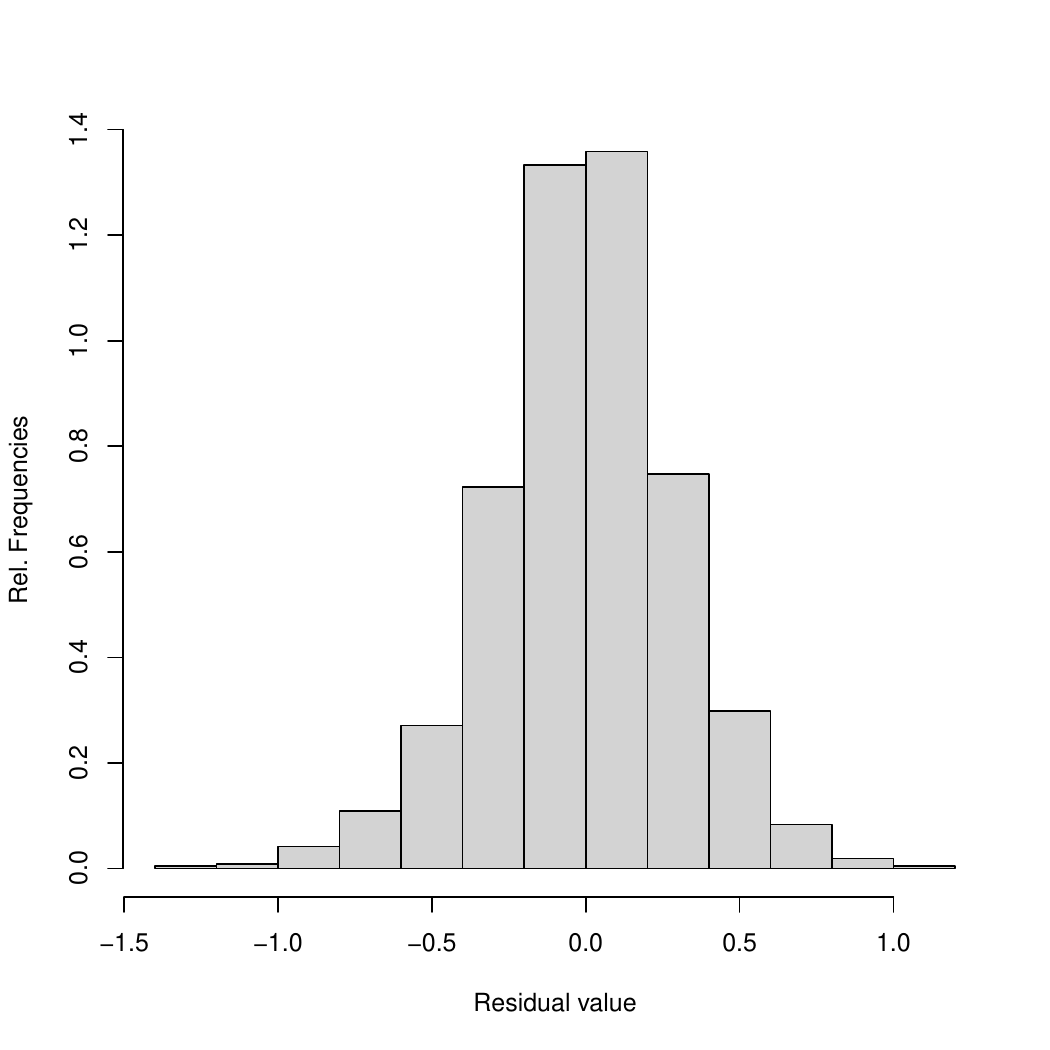}&\includegraphics[width=5cm,height=5.8cm]{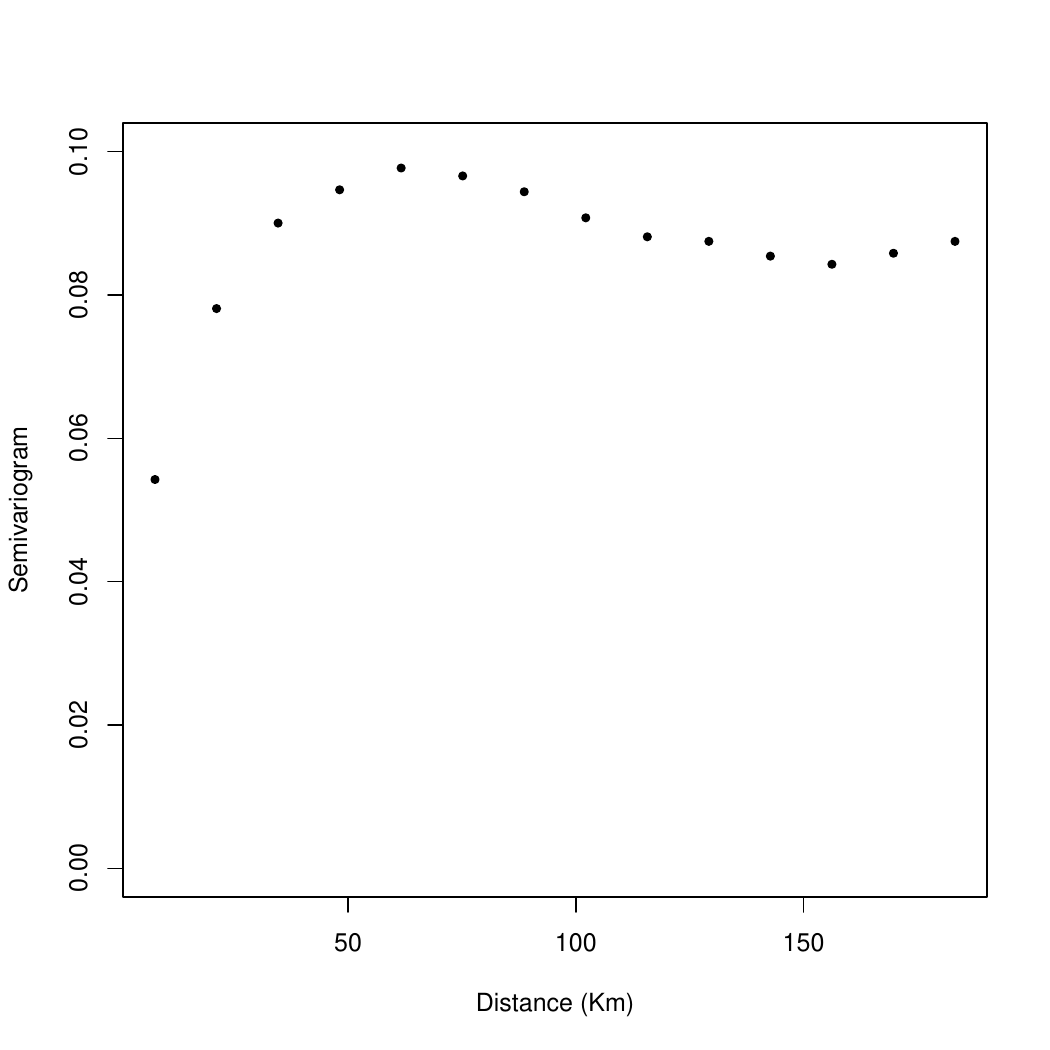}\\
\end{tabular}
\end{center}
\caption{ {From left to right: a color-coded spatial map of the residuals for the Madagascar soil pH data, the corresponding histogram of relative frequencies, and the associated empirical semivariogram.}}
 \label{apps}
\end{figure}

A map of the residuals, the associated histogram and the empirical semivariogram are depicted in Figure \ref{apps}.
It is apparent that a Gaussian random field  with a hole effect isotropic correlation function can be a suitable model in this case. {Indeed, a semivariogram attaining values greater than its sill implies that the covariance function attains negative values. This may reflect cyclic or pseudo-cyclic spatial correlation patterns produced by variations of the soil physico-chemical properties, due to geologic (soil types), geographic (relief), hydrographic (river network), climatic (chemical weathering and leaching), and anthropogenic (deforestation and grazing) factors acting at a scale of a few tens to a hundred kilometers.} As a consequence, we model the residuals as a realization of a zero-mean Gaussian random field in $\mathbb{R}^2$ and we specify the covariance function using the proposed reparameterized hole effect Generalized Wendland model ${\sigma^2 \cal GW}_{\nu a,\xi-{1/2},\nu,2,k}$ for $k=0,1,2,3,{4,5}$.

To evaluate the predictive performances of the different covariance models, we randomly choose 80\% of the spatial locations (2400 data) for estimation and we use the remaining 20\% (600 data) as a validation dataset for predictions.
Table \ref{taba222} reports  detailed statistics on both the estimation and prediction quality for the hole effect reparameterized Generalized Wendland model ${\sigma^2 \cal GW}_{\nu a,\xi-{1/2},\nu,2,k}$ and hole effect Mat\'ern model $\sigma^2{\cal M}_{a,\xi,2,k}$, obtained as a special limit case of ${\sigma^2 \cal GW}_{\nu a,\xi-{1/2},\nu,2,k}$ as $\nu $ tends to $+\infty$ (Section \ref{seccc5}), for $k=1, 2, 3, {4, 5}$. These are:
\begin{enumerate}
\item The maximum likelihood estimates of the model parameters ($\hat{\nu}, \hat{a}, \hat{\sigma}^2, \hat{\xi}$) and the associated standard errors.
\item The Akaike information criterion (AIC) for the estimated models.
\item The root mean squared error (RMSE) and mean absolute value (MAE) obtained when predicting by simple kriging the points included in the validation dataset using the estimated covariance models.
\end{enumerate}

It can be appreciated that the hole effect  reparameterized Generalized Wendland model ${\sigma^2 \cal GW}_{\nu a,\xi-{1/2},\nu,2,k}$ with $k=5$ achieves the lowest AIC and, for a fixed $k$, the AIC criteria always selects this model over the hole effect Mat\'ern model $\sigma^2{\cal M}_{a,\xi,2,k}$. Also, the covariance models with hole effect slightly outperform the covariance models without hole effect in terms of prediction performance. In particular, the predictions obtained using ${\sigma^2 \cal GW}_{\nu a,\xi-{1/2},\nu,2,k}$ with $k=3$ achieve the best RMSE and MAE. The estimated compact support of this best model is $\hat{ \nu} \hat{a}=206.6$ and the associated  covariance matrix has 93\% of zero values, which allows considerably {speeding} the computation of the kriging predictor using algorithms for sparse matrices.

Figure \ref{VVKK}, from left to right, compares the estimated and empirical semivariograms when considering the model ${\sigma^2 \cal GW}_{\nu a,\xi-{1/2},\nu,2,k}$ without and with hole effect ($k=0$ and $k=3$) (left part)
and the model $\sigma^2{\cal M}_{a,\xi,2,k}$ without and with hole effect ($k=0$ and $k=3$) (right part). It is apparent that both models better reproduce the empirical semivariogram when $k=3$, {with an advantage to the $\cal GW$ model, which attains a negative correlation of $-0.065$ (very close to what is observed experimentally), over the $\cal M$ model that only reaches a negative correlation of $-0.020$.} 

\begin{figure}[h!]
\begin{center}
\begin{tabular}{cc}
\includegraphics[width=6.5cm,height=7.2cm]{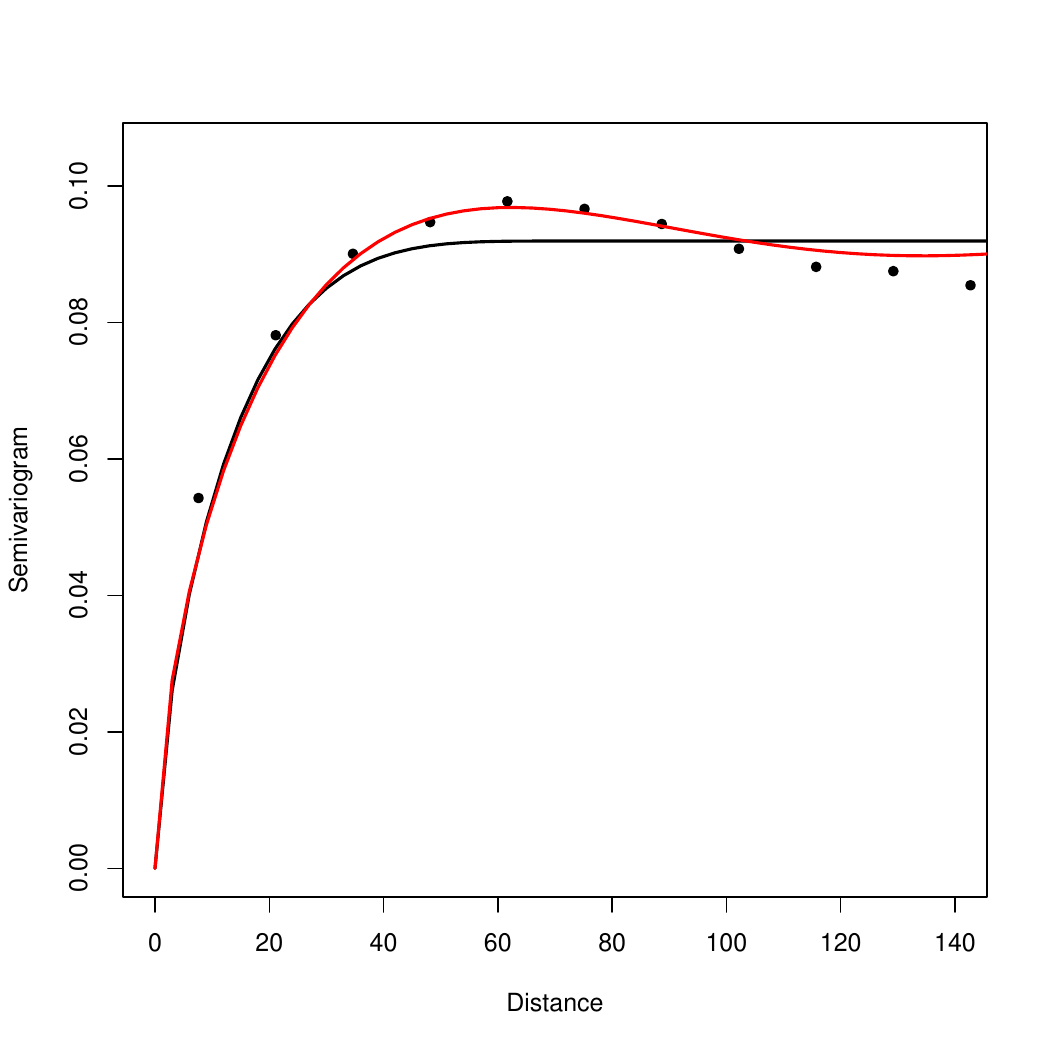}&\includegraphics[width=6.5cm,height=7.2cm]{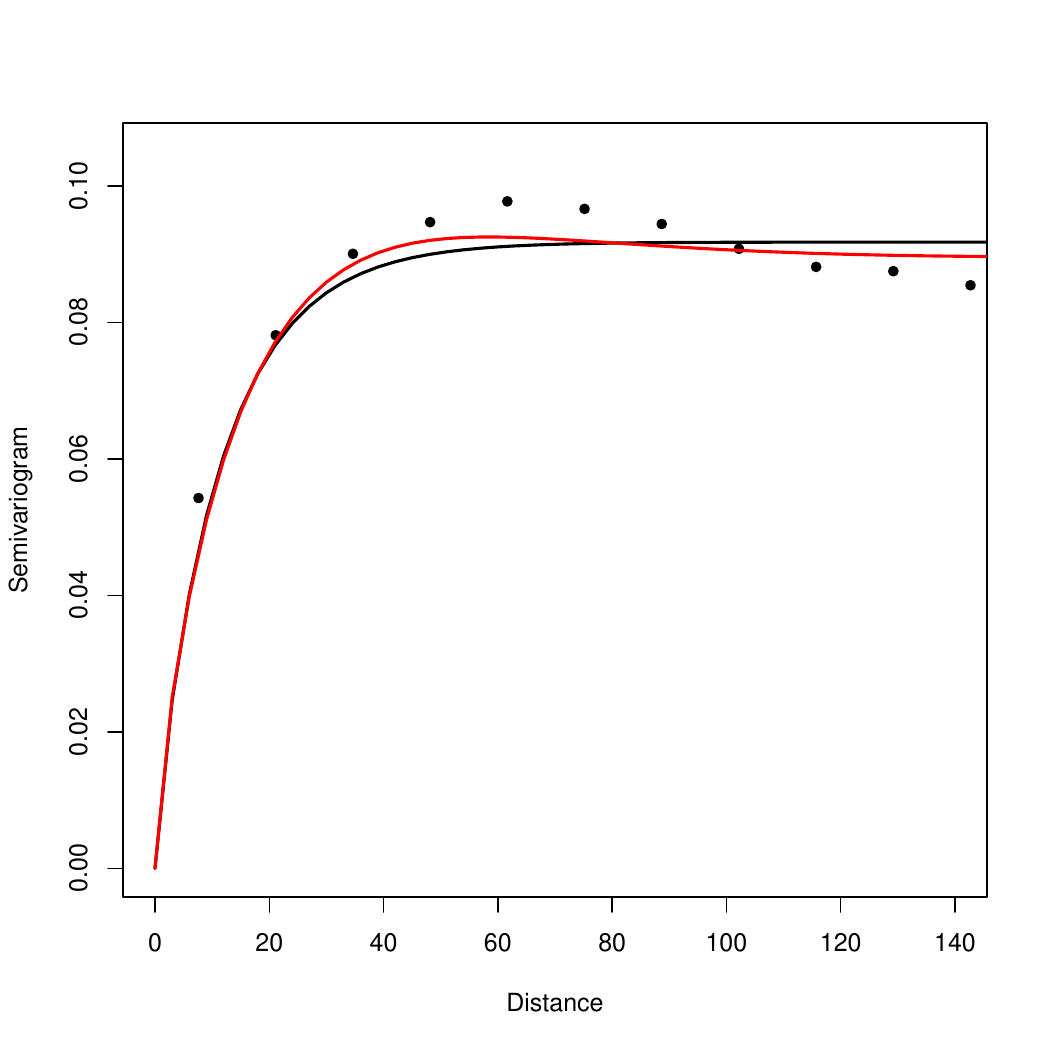}\\
\end{tabular}
\end{center}
\caption{Left: empirical (dots) and estimated (solid lines) semivariograms of the residuals using the hole effect reparameterized Generalized Wendland model ${\sigma^2 \cal GW}_{\nu a,\xi-{1/2},\nu,2,k}$ with $k=0$ (black line) and $k=3$ (red line). Right: the same comparison using the hole effect Mat\'ern model $\sigma^2{\cal M}_{a,\xi,2,k}$ with $k=0$ (black line) and $k=3$ (red line).
}
 \label{VVKK}
\end{figure}

\begin{table}[h!]
\caption{Maximum likelihood estimates, associated standard errors, AIC, and prediction measures (RMSE and MAE) for the hole effect reparameterized Generalized Wendland covariance model ${\sigma^2 \cal GW}_{\nu a,\xi-{1/2},\nu,2,k}$ and the hole effect Mat\'ern covariance model $\sigma^2{\cal M}_{a,\xi,2,k}$, when $k=0,1,2,3,{4,5}$.} 
\label{taba222}
\begin{center}
\scalebox{0.8}{
\begin{tabular}{c|c|c|c|c|c|c|c|}
\cline{2-8}
 & $\hat{\nu}$ & $\hat{a}$ & $\hat{\sigma}^2$ & $\hat{\xi}$ & AIC & RMSE & MAE \\ \hline
\multicolumn{1}{|c|}{\multirow{2}{*}{${\sigma^2 \cal GW}_{\nu a,\xi-{1/2},\nu,2,0}$}} & \multirow{2}{*}{\begin{tabular}[c]{@{}c@{}}$4.289$\\ $(2.541)$\end{tabular}} & \multirow{2}{*}{\begin{tabular}[c]{@{}c@{}}$17.065$\\ $(2.799)$\end{tabular}} & \multirow{2}{*}{\begin{tabular}[c]{@{}c@{}}$0.091$\\ $(0.003)$\end{tabular}} & \multirow{2}{*}{\begin{tabular}[c]{@{}c@{}}$0.336$\\ $(0.046)$\end{tabular}} & \multirow{2}{*}{$410.0$} & \multirow{2}{*}{$0.23750$} & \multirow{2}{*}{$0.17000$} \\
\multicolumn{1}{|c|}{} &  &  &  &  &  &  & \\ \hline
\multicolumn{1}{|c|}{\multirow{2}{*}{$\sigma^2{\cal GW}_{\nu a,\xi-{1/2},\nu,2,1}$}} & \multirow{2}{*}{\begin{tabular}[c]{@{}c@{}}$4.160$\\ $(1.007)$\end{tabular}} & \multirow{2}{*}{\begin{tabular}[c]{@{}c@{}}$29.239$\\ $(2.930)$\end{tabular}} & \multirow{2}{*}{\begin{tabular}[c]{@{}c@{}}$0.091$\\ $(0.003)$\end{tabular}} & \multirow{2}{*}{\begin{tabular}[c]{@{}c@{}}$0.306$\\ $(0.028)$\end{tabular}} & \multirow{2}{*}{$404.8$} & \multirow{2}{*}{$0.23679$} & \multirow{2}{*}{$0.16905$} \\
\multicolumn{1}{|c|}{} &  &  &  &  &  &  & \\ \hline
\multicolumn{1}{|c|}{\multirow{2}{*}{$\sigma^2{\cal GW}_{\nu a,\xi-{1/2},\nu,2,2}$}} & \multirow{2}{*}{\begin{tabular}[c]{@{}c@{}}$4.572$\\ $(0.596)$\end{tabular}} & \multirow{2}{*}{\begin{tabular}[c]{@{}c@{}}$38.053$\\ $(2.976)$\end{tabular}} & \multirow{2}{*}{\begin{tabular}[c]{@{}c@{}}$0.091$\\ $(0.003)$\end{tabular}} & \multirow{2}{*}{\begin{tabular}[c]{@{}c@{}}$0.300$\\ $(0.023)$\end{tabular}} & \multirow{2}{*}{$370.2$} & \multirow{2}{*}{$0.23648$} & \multirow{2}{*}{$0.16866$} \\
\multicolumn{1}{|c|}{} &  &  &  &  &  &  & \\ \hline
\multicolumn{1}{|c|}{\multirow{2}{*}{$\sigma^2{\cal GW}_{\nu a,\xi-{1/2},\nu,2,3}$}} & \multirow{2}{*}{\begin{tabular}[c]{@{}c@{}}$4.514$\\ $(0.132)$\end{tabular}} & \multirow{2}{*}{\begin{tabular}[c]{@{}c@{}}$45.780$\\ $(1.163)$\end{tabular}} & \multirow{2}{*}{\begin{tabular}[c]{@{}c@{}}$0.091$\\ $(0.003)$\end{tabular}} & \multirow{2}{*}{\begin{tabular}[c]{@{}c@{}}$0.295$\\ $(0.017)$\end{tabular}} & \multirow{2}{*}{$362.6$} & \multirow{2}{*}{$0.23630$} & \multirow{2}{*}{$0.16855$} \\
\multicolumn{1}{|c|}{} &  &  &  &  &  &  & \\ \hline
\multicolumn{1}{|c|}{\multirow{2}{*}{$\sigma^2{\cal GW}_{\nu a,\xi-{1/2},\nu,2,4}$}} & \multirow{2}{*}{\begin{tabular}[c]{@{}c@{}}$5.290$\\ $(0.153)$\end{tabular}} & \multirow{2}{*}{\begin{tabular}[c]{@{}c@{}}$50.271$\\ $(1.197)$\end{tabular}} & \multirow{2}{*}{\begin{tabular}[c]{@{}c@{}}$0.091$\\ $(0.003)$\end{tabular}} & \multirow{2}{*}{\begin{tabular}[c]{@{}c@{}}$0.302$\\ $(0.017)$\end{tabular}} & \multirow{2}{*}{$359.9$} & \multirow{2}{*}{$0.23639$} & \multirow{2}{*}{$0.16862$} \\
\multicolumn{1}{|c|}{} &  &  &  &  &  &  & \\ \hline
\multicolumn{1}{|c|}{\multirow{2}{*}{$\sigma^2{\cal GW}_{\nu a,\xi-{1/2},\nu,2,5}$}} & \multirow{2}{*}{\begin{tabular}[c]{@{}c@{}}$6.019$\\ $(0.118)$\end{tabular}} & \multirow{2}{*}{\begin{tabular}[c]{@{}c@{}}$54.832$\\ $(0.975)$\end{tabular}} & \multirow{2}{*}{\begin{tabular}[c]{@{}c@{}}$0.091$\\ $(0.003)$\end{tabular}} & \multirow{2}{*}{\begin{tabular}[c]{@{}c@{}}$0.304$\\ $(0.017)$\end{tabular}} & \multirow{2}{*}{$357.6$} & \multirow{2}{*}{$0.23638$} & \multirow{2}{*}{$0.16858$} \\
\multicolumn{1}{|c|}{} &  &  &  &  &  &  & \\ \hline \hline
\multicolumn{1}{|c|}{\multirow{2}{*}{$\sigma^2{\cal M}_{a,\xi,2,0}$}} & \multirow{2}{*}{\begin{tabular}[c]{@{}c@{}}$-$\\ $ $\end{tabular}} & \multirow{2}{*}{\begin{tabular}[c]{@{}c@{}}$13.094$\\ $(1.491)$\end{tabular}} & \multirow{2}{*}{\begin{tabular}[c]{@{}c@{}}$0.092$\\ $(0.003)$\end{tabular}} & \multirow{2}{*}{\begin{tabular}[c]{@{}c@{}}$0.411$\\ $(0.039)$\end{tabular}} & \multirow{2}{*}{$411.1$} & \multirow{2}{*}{$0.23760$} & \multirow{2}{*}{$0.17034$} \\
\multicolumn{1}{|c|}{} &  &  &  &  &  &  & \\ \hline
\multicolumn{1}{|c|}{\multirow{2}{*}{$\sigma^2{\cal M}_{a,\xi,2,1}$}} & \multirow{2}{*}{\begin{tabular}[c]{@{}c@{}}$-$\\ $ $\end{tabular}} & \multirow{2}{*}{\begin{tabular}[c]{@{}c@{}}$21.595$\\ $(1.811)$\end{tabular}} & \multirow{2}{*}{\begin{tabular}[c]{@{}c@{}}$0.090$\\ $(0.003)$\end{tabular}} & \multirow{2}{*}{\begin{tabular}[c]{@{}c@{}}$0.386$\\ $(0.030)$\end{tabular}} & \multirow{2}{*}{$388.9$} & \multirow{2}{*}{$0.23701$} & \multirow{2}{*}{$0.16956$} \\
\multicolumn{1}{|c|}{} &  &  &  &  &  &  & \\ \hline
\multicolumn{1}{|c|}{\multirow{2}{*}{$\sigma^2{\cal M}_{a,\xi,2,2}$}} & \multirow{2}{*}{\begin{tabular}[c]{@{}c@{}}$-$\\ $ $\end{tabular}} & \multirow{2}{*}{\begin{tabular}[c]{@{}c@{}}$28.281$\\ $(2.153)$\end{tabular}} & \multirow{2}{*}{\begin{tabular}[c]{@{}c@{}}$0.090$\\ $(0.003)$\end{tabular}} & \multirow{2}{*}{\begin{tabular}[c]{@{}c@{}}$0.376$\\ $(0.028)$\end{tabular}} & \multirow{2}{*}{$381.5$} & \multirow{2}{*}{$0.23682$} & \multirow{2}{*}{$0.16924$} \\
\multicolumn{1}{|c|}{} &  &  &  &  &  &  & \\ \hline
\multicolumn{1}{|c|}{\multirow{2}{*}{$\sigma^2{\cal M}_{a,\xi,2,3}$}} & \multirow{2}{*}{\begin{tabular}[c]{@{}c@{}}$-$\\ $ $\end{tabular}} & \multirow{2}{*}{\begin{tabular}[c]{@{}c@{}}$33.776$\\ $(2.454)$\end{tabular}} & \multirow{2}{*}{\begin{tabular}[c]{@{}c@{}}$0.090$\\ $(0.003)$\end{tabular}} & \multirow{2}{*}{\begin{tabular}[c]{@{}c@{}}$0.371$\\ $(0.027)$\end{tabular}} & \multirow{2}{*}{$378.2$} & \multirow{2}{*}{$0.23677$} & \multirow{2}{*}{$0.16913$} \\
\multicolumn{1}{|c|}{} &  &  &  &  &  &  & \\ \hline
\multicolumn{1}{|c|}{\multirow{2}{*}{$\sigma^2{\cal M}_{a,\xi,2,4}$}} & \multirow{2}{*}{\begin{tabular}[c]{@{}c@{}}$-$\\ $ $\end{tabular}} & \multirow{2}{*}{\begin{tabular}[c]{@{}c@{}}$38.512$\\ $(2.720)$\end{tabular}} & \multirow{2}{*}{\begin{tabular}[c]{@{}c@{}}$0.090$\\ $(0.003)$\end{tabular}} & \multirow{2}{*}{\begin{tabular}[c]{@{}c@{}}$0.369$\\ $(0.026)$\end{tabular}} & \multirow{2}{*}{$376.3$} & \multirow{2}{*}{$0.23676$} & \multirow{2}{*}{$0.16907$} \\
\multicolumn{1}{|c|}{} &  &  &  &  &  &  & \\ \hline
\multicolumn{1}{|c|}{\multirow{2}{*}{$\sigma^2{\cal M}_{a,\xi,2,5}$}} & \multirow{2}{*}{\begin{tabular}[c]{@{}c@{}}$-$\\ $ $\end{tabular}} & \multirow{2}{*}{\begin{tabular}[c]{@{}c@{}}$42.752$\\ $(2.961)$\end{tabular}} & \multirow{2}{*}{\begin{tabular}[c]{@{}c@{}}$0.0897 $\\ $(0.003)$\end{tabular}} & \multirow{2}{*}{\begin{tabular}[c]{@{}c@{}}$0.367$\\ $(0.026)$\end{tabular}} & \multirow{2}{*}{$375.1$} & \multirow{2}{*}{$0.23675$} & \multirow{2}{*}{$0.16904$} \\
\multicolumn{1}{|c|}{} &  &  &  &  &  &  & \\ \hline
\end{tabular}
}
\end{center}
\end{table}

\section{Concluding remarks}
\label{sec:concl}

This paper generalizes the Mat\'ern  and Generalized Wendland isotropic correlation models, allowing them to attain negative values. An additional positive integer parameter describes the amplitude of negative correlations. As a result, the proposed correlation models are very flexible since they allow {jointly parameterizing} smoothness, global or compact support, and hole effects.

The proposed generalizations depend on the evaluation of some special functions, which can be performed through efficient implementation of the Bessel and Gauss hypergeometric functions as in the R package \texttt{GeoModels}  \citep{Bevilacqua:2018aa}. However,
the computation greatly simplifies in some important special cases, which makes them attractive to practitioners.

Numerical evidences of the versatility of these hole effect models have been provided. In particular, the real data application showed how accounting for a hole effect can improve goodness of fit and prediction accuracies. In addition, under an  increasing domain scenario, the results of a simulation study suggest that the hole effect does not affect the estimation of the covariance parameters. A topic for future work is the estimation of the proposed models under fixed domain asymptotics, which implies the study of the equivalence of Gaussian measures for these models.
{
It is well known that the Matérn and Generalized Wendland isotropic correlation models are compatible in dimensions 
$d=1,2,3$ under specific conditions \citep{BFFP}. It would be interesting to verify whether this compatibility still holds when using the hole-effect generalizations proposed in this paper.}

Another interesting topic for future work is the comparison, from a statistical point of view, of the two different parameterizations (\ref{ppoi}) and (\ref{ppoi4}) of the Generalized Wendland model that include the Mat\'ern model as a special limit case.

\section{{Appendix}}
\label{suppmaterial}

The {Appendix} is outlined as follows. Section \ref{appendA} contains additional notation. Sections \ref{app:TBoperator} and \ref{app:montee} recall two operators (the turning bands and mont\'ee, respectively) that transform a covariance defined in a $d$-dimensional Euclidean space to another covariance defined in a space of lower or higher dimension. Section \ref{app:lemmas} gives two useful lemmas. Sections \ref{alternatives_Matern} and \ref{alternatives} provide alternative analytical expressions of the hole effect Mat\'ern and hole effect Generalized Wendland models, respectively. Finally, Section \ref{app:proofs} contains the proofs of the propositions stated in this paper.

\subsection{Additional Notation}
\label{appendA}

In what follows, $\mathbb{N}_{\geq \alpha}$ denotes the set of integers greater than, or equal to, $\alpha$, and $\mathbb{R}_{> \alpha}$ is the set of real numbers greater than $\alpha$; $\mathbb{N}_{\geq 0}$ is simply denoted as $\mathbb{N}$. Also, $\mathrm{i}$ is the imaginary unit.

The functions in Table \ref{tab:special function} will be used in this {Appendix}.
\begin{table}[h!]
    \caption{Special functions  \citep{exton, Olver}.}
    \label{tab:special function}
    \begin{tabular}{ c  l l }
        \hline
        Notation & Function name & Parameters \\
        \hline
        $\lceil \cdot \rceil$ & Ceil function\\
        $\Gamma(\cdot,\cdot)$ & Incomplete Gamma function &\\
        $I_{\nu}$ & Modified Bessel function of the first kind & $\nu \in \mathbb{R}$\\
        $P_{\nu}^\mu$ & Associated Legendre function of the first kind & $\mu, \nu \in \mathbb{R}$\\
        $Q_{\nu}^\mu$ & Associated Legendre function of the second kind & $\mu, \nu \in \mathbb{R}$\\
        $G_{p,q}^{m,n}({\boldsymbol{\beta};\boldsymbol{\gamma};\cdot})$ & Meijer-$G$ function & $p,q \in \mathbb{N}$, $\boldsymbol{\beta} \in \mathbb{R}^p$, $\boldsymbol{\gamma} \in \mathbb{R}^q$        \\
        \hline
    \end{tabular}
\end{table}

\subsection{The Turning Bands Operator}
\label{app:TBoperator}

The turning bands operator connects members of the class $\Phi_{d+2}$ with members of the class  $\Phi_{d}$ and can be summarized in the following lemma \citep[p. 21]{Matheron1972}.
\begin{lemma}
\label{TB2}
Let $d \in \mathbb{N}_{\geq 1}$ and $C \in \Phi_{d+2}$. The turning bands operator ${\mathfrak T}_{d+2,d}$ applied to $C$ is the mapping defined by
\begin{equation}
\label{turning03}
      {\mathfrak T}_{d+2,d}\left[ C\right](h) = \frac{h^{1-d}}{d} \frac{\partial [h^{d} C(h)]}{\partial h} , \quad h > 0,
\end{equation}
which belongs to $\Phi_{d}$.
\end{lemma}

It is well-known that the turning bands operator preserves the local behavior of the correlation function at the  origin and allows {attaining} negative value at same time \citep{gneiting2002compactly}. This implies that a natural approach to obtain a Mat\'ern or Generalized Wendland model with hole effect is an application, or more generally a recursive application, of the turning bands operator, which will be done in the next sections of this {Appendix}.

With this goal in mind, the following lemma provides a general expression for the correlation model obtained by applying $k$ times the turning bands operator and is of independent interest.

\begin{lemma}
\label{TB2k}
Let $k \in \mathbb{N}$, $d \in \mathbb{N}_{\geq 1}$ and $C \in \Phi_{d+2k}$. A recursive application of the turning bands operator applied to $C$ gives the following mapping, which belongs to $\Phi_d$:
\begin{equation}
\label{explicitTB}
\begin{split}
{\mathfrak T}_{d+2k,d}\left[C\right](h) :&= {\mathfrak T}_{d+2,d} \circ \ldots \circ {\mathfrak T}_{d+2(k-1),d+2(k-2)} \circ {\mathfrak T}_{d+2k,d+2(k-1)} \left[C\right](h)\\
&= \sum_{q=0}^k \sum_{r=0}^{\max\{0,q-1\}} \frac{(-1)^r (k-q+1)_q (q)_r (q-r)_r \, h^{q-r} \, C^{(q-r)}(h)}{2^{q+r} q! \, r! (\frac{d}{2})_q}, \quad h>0,
\end{split}
\end{equation}
where ${\mathfrak T}_{d+2,d} \circ \ldots \circ {\mathfrak T}_{d+2(k-1),d+2(k-2)} \circ {\mathfrak T}_{d+2k,d+2(k-1)} \left[C\right]$ is the $k$-fold composition of $C$ recursively by

 $${\mathfrak T}_{d+2k,d+2(k-1)}, {\mathfrak T}_{d+2(k-1),d+2(k-2)} \ldots,{\mathfrak T}_{d+2,d}.$$ Moreover, if $C$ has a $(d+2k)$-radial density $\widehat{C}$, then ${\mathfrak T}_{d+2k,d}\left[C\right]$ has a $d$-radial density given by
\begin{equation}
\label{tb2kdensity}
\widehat{{\mathfrak T}_{d+2k,d}\left[C\right]}(u) := \frac{\pi^{k}}{(\frac{d}{2})_k} \, u^{2k} \, \widehat{C}(u), \quad u \geq 0.
\end{equation}
\end{lemma}

\begin{proof}[Proof of Lemma \ref{TB2k}]
The fact that ${\mathfrak T}_{d+2k,d}\left[ C\right]$ belongs to $\Phi_d$ is a consequence of Lemma \ref{TB2}. Assume that $k$ is non-zero. To get an analytical expression of this mapping, we invoke formula 5.4'' of \cite{Matheron1972} to write
\begin{equation*}
  C(h) = \frac{2 \Gamma(\frac{d}{2}+k)}{ \Gamma(k) \Gamma(\frac{d}{2})} h^{2-d-2k} \int_0^{h} u^{d-1} (h^2-u^2)^{k-1} {\mathfrak T}_{d+2k,d}\left[C\right](u) {\rm d}u, \quad h > 0.
\end{equation*}
Equivalently,
\begin{equation*}
  x^{{d/2}+k-1} C(\sqrt{x}) = \frac{2 \Gamma(\frac{d}{2}+k)}{ \Gamma(k) \Gamma(\frac{d}{2})} \int_0^{\sqrt{x}} u^{d-1} (x-u^2)^{k-1} {\mathfrak T}_{d+2k,d}\left[C\right](u) {\rm d}u, \quad h > 0.
\end{equation*}
Differentiating $k$ times this identity and using \cite[0.410, 0.42 and 0.433.1]{grad} gives (\ref{explicitTB}), which remains valid for $k=0$. Concerning the spectral densities, (\ref{tb2kdensity}) stems from formula 5.3 of \cite{Matheron1972}.
\end{proof}

\subsection{The Mont\'ee and Descente Operators}
\label{app:montee}

Let $p \in \mathbb{N}$ and $C \in \Phi_{d+p}$ such that $C(\| \cdot \|_{d+p})$ is absolutely integrable in $\mathbb{R}^{d+p}$. Then, it possesses a $(d+p)$-radial spectral density $\widehat{C}_{d+p}$. Its \emph{mont\'ee} (\emph{upgrading}) of order $p$ is the member of $\Phi_{d}$ (denoted as ${\mathfrak M}_{d+p,d}[C]$) with $d$-radial spectral density $f_{d+p}$. This translates into the integral relation  \citep[I.4.18]{matheron1965}
\begin{equation*}
      {\mathfrak M}_{d+p,d}[C](h) = \frac{2\pi^{{p/2}}}{\Gamma(\frac{p}{2})} \int_h^{+\infty} u (u^2-h^2)^{{p/2}-1} C(u) {\rm d}u, \quad h \geq 0.
\end{equation*}
This equation can be taken as the definition of a mont\'ee of fractional order when $p$ is not an integer.

The descente operator is defined as the reciprocal of the mont\'ee. Let $C$ belong to $\Phi_{d-p}$ with $(d-p)$-radial spectral density $\widehat{C}_{d-p}$ such that $\widehat{C}_{d-p}(\| \cdot \|_{d})$ is absolutely integrable in $\mathbb{R}^{d}$. Its \emph{descente} (\emph{downgrading}) of order $p$ is the member of $\Phi_{d}$ (denoted as ${\mathfrak M}_{d-p,d}[C]$) with $d$-radial spectral density $f_{d-p}$. If $p$ is not an integer, one can define the descente of fractional order $p$ by a descente of order $\lceil p \rceil$ followed by a mont\'ee of order $\lceil p \rceil - p$ \citep{matheron1965}.

\subsection{Other Useful Lemmas}
\label{app:lemmas}

\begin{lemma}[Dini's second theorem]
\label{Dini}
Let $\{ f_n : n \in \mathbb{N} \}$ be a sequence of real-valued non-increasing functions on $[0,+\infty)$ that converges pointwise to a continuous function $f$. Then, the convergence is uniform on $[0,+\infty)$.
\end{lemma}

\begin{proof}
See \cite[p. 81]{Polya-Szego:1998}.
\end{proof}

\begin{lemma}
\label{monotonicity}
Let $C \in \Phi_d$ with $d$-radial spectral density $\widehat{C}_d$. Then $C \in \Phi_{d+2}$ if, and only if, $\widehat{C}_d$ is non-increasing on $[0,+\infty)$.
\end{lemma}

\begin{proof}
See \cite[p. 497]{gneiting2002compactly}.
\end{proof}

\subsection{Alternative Analytical Expressions of the Hole Effect Mat\'ern Model}
\label{alternatives_Matern}

We provide alternative expressions of ${\cal M}_{a,\xi,d,k}$ involving special functions.

\begin{enumerate}
\item[1.] \textbf{Expression in terms of generalized hypergeometric functions.} The $(d,k)$-hole effect Mat\'ern model is the Fourier-Hankel transformation (\ref{fourier1}) of its $d$-radial spectral density (\ref{stein2}), $i.e.$:
\begin{equation}
\label{generalizedMatern0}
    {\cal M}_{a,\xi,d,k}(h) = \frac{2^{{d/2}}\Gamma(\xi+\frac{d}{2}+k) \Gamma(\frac{d}{2}) a^{d+2k}}{\Gamma(\frac{d}{2}+k) \Gamma(\xi)} h^{1-{d/2}} \int_0^{+\infty} \frac{u^{2k+{d/2}} J_{{d/2}-1}\left(\frac{uh}{a}\right)}{(1+a^2 u^2)^{\xi+{d/2}+k}} {\rm d}u.
\end{equation}
For $\xi \notin \mathbb{N}_{\geq 1}$, this leads to \cite[6.565.8]{grad}
\begin{equation}
\label{generalizedMatern}
    \begin{split}
    {\cal M}_{a,\xi,d,k}(h)
    &= \frac{\Gamma(\xi+\frac{d}{2}+k) \Gamma(\frac{d}{2}) \Gamma(-\xi)}{\Gamma(\frac{d}{2}+k) \Gamma(\xi)\Gamma(\xi+\frac{d}{2})} \left(\frac{h}{2a}\right)^{2\xi} {{}_1F_2\left(\xi+\frac{d}{2}+k;\xi+\frac{d}{2},\xi+1;\frac{h^2}{4a^2}\right)} \\
    &+{{}_1F_2\left(k+\frac{d}{2};1-\xi,\frac{d}{2};\frac{h^2}{4a^2}\right)}, \quad h \geq 0, \xi \notin \mathbb{N}_{\geq 1}.
    \end{split}
\end{equation}
\item[2.] \textbf{Expression in terms of modified Bessel functions of the first kind.} Using a Kummer-type transformation \citep[Theorem 2.1]{Withers}, one can rewrite (\ref{generalizedMatern}) as
\begin{equation*}
    \begin{split}
    {\cal M}_{a,\xi,d,k}(h) &= \sum_{n=0}^k \frac{k!}{n!(k-n)! (1-\xi)_n (\frac{d}{2})_n} {{}_0F_1\left(;1-\xi+n;\frac{h^2}{4a^2}\right)} \left(\frac{h^2}{4a^2}\right)^n \\
    &+ \frac{\Gamma(\xi+\frac{d}{2}+k) \Gamma(\frac{d}{2}) \Gamma(-\xi)}{\Gamma(k+\frac{d}{2}) \Gamma(\xi)\Gamma(\xi+\frac{d}{2})} \left(\frac{h}{2a}\right)^{2\xi} \\
    &\times \sum_{n=0}^k \frac{k!}{n!(k-n)! (\xi+1)_n (\xi+\frac{d}{2})_n} {{}_0F_1\left(;\xi+1+n;\frac{h^2}{4a^2}\right)} \left(\frac{h^2}{4a^2}\right)^n, \quad h \geq 0, \xi \notin \mathbb{N}_{\geq 1}.
    \end{split}
\end{equation*}

In turn, the ${}_0F_1$ hypergeometric function can be expressed in terms of modified Bessel functions of the first kind \citep[10.39.9]{Olver}, which leads to:
\begin{equation*}
    \begin{split}
    {\cal M}_{a,\xi,d,k}(h) &= \sum_{n=0}^k \frac{k! \Gamma(\frac{d}{2}) \Gamma(1-\xi)}{n!(k-n)!} \left(\frac{h}{2a}\right)^{n+\xi} \\
    &\times \left[\frac{I_{n-\xi}\left(\frac{h}{a}\right)}{\Gamma(\frac{d}{2}+n)} - \frac{\Gamma(\xi+\frac{d}{2}+k) I_{n+\xi}\left(\frac{h}{a}\right)}{\Gamma(k+\frac{d}{2}) \Gamma(\xi+\frac{d}{2}+n)}\right], \quad h \geq 0, \xi \notin \mathbb{N}_{\geq 1}.
    \end{split}
\end{equation*}
If $k=0$, then one recovers the traditional Mat\'ern correlation (\ref{Matern}) owing to formulae 5.5.3 and 10.27.4 of \cite{Olver}, that is ${\cal M}_{a,\xi,d,0}= {\cal M}_{a,\xi}$.

\item[3.] \textbf{Expression in terms of a Meijer function.} The integral in (\ref{generalizedMatern0}) can be expressed by means of a Meijer-$G$ function, which leads to [\citealp[14.4.21]{erdelyi1954b}; \citealp[8.2.2.15]{prud}]:
\begin{equation*}
    {\cal M}_{a,\xi,d,k}(h) =
    \begin{cases}
    \frac{\Gamma(\frac{d}{2})}{\Gamma(k+{d/2}) \Gamma(\xi)} {G_{1,3}^{2,1}\left(1-k-\frac{d}{2};\xi,0,1-\frac{d}{2};\frac{h^2}{4a^2}\right)} 
    \text{ if $h > 0$}\\
    1 \text{ if $h = 0$.}
    \end{cases}
\end{equation*}
The function on the right-hand side is well defined for all $\xi>0$, showing that ${\cal M}_{a,\xi,d,k}$ can be continued to a member of $\Phi_d$ when $\xi \in \mathbb{N}$.

\end{enumerate}

\subsection{Alternative Analytical Expressions of the Hole Effect Generalized Wendland Model}
\label{alternatives}

Let $a \in \mathbb{R}_{> 0}$, $\xi \in \mathbb{R}_{>-{1/2}}$ and $k \in \mathbb{N}$ satisfying conditions (A) to (C) of Proposition \ref{genW}.
We provide alternative expressions of ${\cal GW}_{a,\xi,\nu,d,k}$ involving special functions.

\begin{enumerate}
\item[1.] \textbf{Expressions in terms of Gauss hypergeometric functions.}
Using formula 9.6.5 of \cite{lebedev65}, one can rewrite (\ref{wendlandext}) as
\begin{equation*}
\label{genordwend3}
\begin{split}
    &{\cal GW}_{a,\xi,\nu,d,k}(h)
    = \sum_{n=0}^k \frac{(-1)^n k! (\frac{1-\nu}{2}-\xi)_n (-\xi-\frac{\nu}{2})_n}{n! (k-n)! (1-\frac{d}{2}-n)_n (\frac{1}{2}-\xi)_n} \left(\frac{h}{a}\right)^{2n} \\
    & \times \left(\frac{1+\sqrt{1-\frac{h^2}{a^2}}}{2}\right)^{2\xi+\nu-2n} {{}_2F_1\left(2n-2\xi-\nu,\frac{1}{2}-\xi-\nu+n;\frac{1}{2}-\xi+n;\frac{1-\sqrt{1-\frac{h^2}{a^2}}}{1+\sqrt{1-\frac{h^2}{a^2}}}\right)}\\
    &+ \frac{\Gamma(\xi+\frac{1+\nu}{2}) \Gamma(\xi+\frac{\nu}{2}+1) \Gamma(\frac{d}{2})\Gamma(-\xi-\frac{1}{2})}{\Gamma(\frac{d}{2}+k) \Gamma(\xi+\frac{1}{2})\Gamma(\frac{\nu}{2})\Gamma(\frac{\nu+1}{2})}\\
    & \times \sum_{n=0}^k \frac{(-1)^{n+k} k! (\frac{1-d}{2}-\xi-k)_{k-n} (1-\frac{\nu}{2})_n (\frac{1-\nu}{2})_n}{n! (k-n)! (\xi+\frac{3}{2})_n} \left(\frac{h}{a}\right)^{2\xi+1+2n}\\
    & \times \left(\frac{1+\sqrt{1-\frac{h^2}{a^2}}}{2}\right)^{\nu-1-2n} {{}_2F_1\left(1-\nu+2n,n-\nu-\xi+\frac{1}{2};\xi+\frac{3}{2}+n;\frac{1-\sqrt{1-\frac{h^2}{a^2}}}{1+\sqrt{1-\frac{h^2}{a^2}}}\right)}, \quad 0 \leq h < a.
\end{split}
\end{equation*}
One can also use a quadratic transformation \citep[15.8.24]{Olver} of the arguments in the Gauss hypergeometric functions to obtain
\begin{equation*}
\begin{split}
    &{\cal GW}_{a,\xi,\nu,d,k}(h)
    \\&= \sum_{n=0}^k \frac{(-1)^n \Gamma(\frac{1}{2}) \Gamma(\frac{1}{2}-\xi-n) k!}{n! (k-n)! (1-\frac{d}{2}-n)_n (\frac{1}{2}-\xi)_n} \left(\frac{h}{a}\right)^{2n} \left(1-\frac{h^2}{a^2}\right)^{\xi+{(\nu-1)/2}-n} \\
    &\times \Bigg[\frac{(-\xi-\frac{\nu}{2})_n }{\Gamma(\frac{1-\nu}{2}-\xi)\Gamma(\frac{1+\nu}{2})} \left(1-\frac{h^2}{a^2}\right)^{{1/2}} {{}_2F_1\left(n-\xi-\frac{\nu}{2},\frac{\nu}{2};\frac{1}{2};\frac{1}{1-\frac{h^2}{a^2}}\right)}\\
    &-\frac{2(\frac{1-\nu}{2}-\xi)_n}{\Gamma(-\xi-\frac{\nu}{2})\Gamma(\frac{\nu}{2})}     {{}_2F_1\left(n-\xi+\frac{1-\nu}{2},\frac{1+\nu}{2};\frac{3}{2};\frac{1}{1-\frac{h^2}{a^2}}\right)}\Bigg]
    \\
    &- \sum_{n=0}^k \frac{(-1)^{n+k} \Gamma(\frac{1}{2}-\xi) k! (\frac{1-d}{2}-\xi-k)_{k-n} (1-\frac{\nu}{2})_n (\frac{1-\nu}{2})_n}{2^{1-\nu} (\frac{d}{2})_k \Gamma(\nu) n! (k-n)!} \left(\frac{h}{a}\right)^{2\xi+1+2n} \\
    & \times \left(1-\frac{h^2}{a^2}\right)^{{\nu/2}-1-n} \Bigg[ \frac{\Gamma(\xi+\frac{1+\nu}{2})}{\Gamma(1-\frac{\nu}{2}+n)} \left(1-\frac{h^2}{a^2}\right)^{{1/2}} {{}_2F_1\left(\frac{1-\nu}{2}+n,\xi+\frac{1+\nu}{2};\frac{1}{2};\frac{1}{1-\frac{h^2}{a^2}}\right)}\\
    &- \frac{2 \Gamma(\xi+1+\frac{\nu}{2})}{\Gamma(\frac{1-\nu}{2}+n)} {{}_2F_1\left(1-\frac{\nu}{2}+n,\xi+\frac{\nu}{2}+1;\frac{3}{2};\frac{1}{1-\frac{h^2}{a^2}}\right)}\Bigg], \quad 0 \leq h < a.
\end{split}
\end{equation*}
Alternatively, one can use the fact that ${\cal GW}_{a,\xi,\nu,d,k}$ is obtained from ${\cal GW}_{a,\xi,\nu,d+2k,0}$ by applying $k$ times the turning bands operator (see (\ref{meijer}) in the proof of Proposition \ref{genW}). Using Lemma \ref{TB2k}, (\ref{WG4*}) and formula 0.432 of \cite{grad}, one finds
\begin{equation*}
\label{GeneralizedHypergeometric4}
\begin{split}
    {\cal GW}_{a,\xi,\nu,d,k} &=\sum_{q=0}^k \sum_{r=0}^q \sum_{s=0}^{\max \{0,q-r-1\}} \frac{(-1)^q (k-q+1)_q (q)_r (q-r)_r (q-r-2s+1)_{2s}}{2^{2r+2s} q! \, r! \, s! \, (\frac{d}{2})_q} \\
    &\quad \times \frac{\Gamma(\xi+\frac{1+\nu}{2})\Gamma(\xi+\frac{\nu}{2}+1) }{\Gamma(\xi+\nu+1-q+r+s) \Gamma(\xi+\frac{1}{2})} \\
    &\quad \times \left(\frac{h}{a}\right)^{q-r} \left(1-\frac{h}{a}\right)_+^{\xi+\nu-s} \left(1+\frac{h}{a}\right)^{\xi+\nu-q+r+s} \\
    &\quad \times {{}_2F_1\left(\frac{\nu}{2},\frac{\nu+1}{2};\xi+\nu+1-q+r+s;1-\frac{h^2}{a^2}\right)}, \quad 0 < h \leq a.
\end{split}
\end{equation*}
The latter expression is well-defined and is continuous on $(0,a]$ and vanishes at $h=a$, even when condition (C) does not hold, which proves that ${\cal GW}_{a,\xi,\nu,d,k}$ can be continued when $\xi$ is a half-integer. The continuation so obtained is still a member of $\Phi_d$ insofar as ${\cal GW}_{a,\xi,\nu,d,k} = {\mathfrak T}_{d+2k,d}[{\cal GW}_{a,\xi,\nu,d+2k,0}]$ with ${\cal GW}_{a,\xi,\nu,d+2k,0} \in \Phi_{d+2k}$.
\item[2.] \textbf{Expressions in terms of associated Legendre functions.}
Using formulae 7.3.1.100 and 7.3.1.102 of \cite{prud}, one can express the hypergeometric functions in (\ref{wendlandext}) in terms of associated Legendre functions of the first or second kind. This gives:
\begin{equation*}
\begin{split}
    &{\cal GW}_{a,\xi,\nu,d,k}(h)
    = \sum_{n=0}^k \frac{(-1)^n k! (-\xi+\frac{1-\nu}{2})_n (-\xi-\frac{\nu}{2})_n \Gamma(1-\xi-\frac{1}{2})}{2^{\xi+{1/2}-n} n! (k-n)! (1-\frac{d}{2}-n)_n} \\
    &\qquad \times \left(\frac{h}{a}\right)^{\xi+{1/2}+n}  \left(1-\frac{h^2}{a^2}\right)^{{(\xi+\nu-n)/2}-{1/4}} P^{\xi+{1/2}-n}_{-\xi-{1/2}-\nu+n}\left(\frac{1}{\sqrt{1-\frac{h^2}{a^2}}}\right)\\
    &- \frac{\Gamma(2\xi+1+\nu) \Gamma(\frac{d}{2})\Gamma(\frac{1}{2}-\xi)}{\Gamma(\frac{d}{2}+k) \Gamma(\nu)} \sum_{n=0}^k \frac{(-1)^{n+k} k! (\frac{1-d}{2}-\xi-k)_{k-n} (1-\frac{\nu}{2})_n (\frac{1-\nu}{2})_n}{2^{\xi+{1/2}-n} n! (k-n)!}\\
    &\qquad \times \left(\frac{h}{a}\right)^{\xi+{1/2}+n} \left(1-\frac{h^2}{a^2}\right)^{{(\xi+\nu-n)/2}-{1/4}} P^{-\xi-{1/2}-n}_{-\xi-{1/2}-\nu+n}\left(\frac{1}{\sqrt{1-\frac{h^2}{a^2}}}\right), \quad 0 \leq h < a,
\end{split}
\end{equation*}
and
\begin{equation*}
\begin{split}
    &{\cal GW}_{a,\xi,\nu,d,k}(h)
    = \sum_{n=0}^k \frac{(-1)^n k! (-\xi+\frac{1-\nu}{2})_n (-\xi-\frac{\nu}{2})_n \Gamma(1-\xi-\frac{1}{2})}{2^{\xi-n} \sqrt{\pi} n! (k-n)! (1-\frac{d}{2}-n)_n \Gamma(\nu)} \\
    &\qquad \times \left(\frac{h}{a}\right)^{\xi+n}  \left(1-\frac{h^2}{a^2}\right)^{{(\xi+\nu-n)/2}} e^{(n-\xi-\nu)\mathrm{i}\pi} Q_{-\xi-1+n}^{\xi+\nu-n}\left(\frac{a}{h}\right)\\
    &- \frac{\Gamma(\frac{d}{2})\Gamma(\frac{1}{2}-\xi)}{\Gamma(\frac{d}{2}+k) \Gamma(\nu)} \sum_{n=0}^k \frac{(-1)^{n+k} k! (\frac{1-d}{2}-\xi-k)_{k-n} (1-\frac{\nu}{2})_n (\frac{1-\nu}{2})_n}{2^{\xi-n} \sqrt{\pi} n! (k-n)!}\\
    &\qquad \times \left(\frac{h}{a}\right)^{\xi+n} \left(1-\frac{h^2}{a^2}\right)^{{(\xi+\nu-n)/2}} e^{(n-\xi-\nu)\mathrm{i}\pi} Q_{\xi+n}^{\xi+\nu-n}\left(\frac{a}{h}\right), \quad 0 \leq h < a.
\end{split}
\end{equation*}
\item[3.] \textbf{Expression in terms of a Meijer function.}
Using (\ref{wendlandext0}) and \cite[8.2.2.3 and 8.2.2.15]{prud}, one obtains
\begin{equation}
\label{GW2Meijer}
    \begin{split}
    &{\cal GW}_{a,\xi,\nu,d,k}(h) \\&= \begin{cases}
    0 \text{ if $a \leq h$}\\
    \frac{\Gamma(\frac{d}{2}) \Gamma(\xi+\frac{1+\nu}{2}) \Gamma(\xi+\frac{\nu}{2}+1)}{\Gamma(\xi+\frac{1}{2}) \Gamma(\frac{d}{2}+k)} {G^{2,1}_{3,3}\left(1-\frac{d}{2}-k,\xi+\frac{1+\nu}{2},\xi+\frac{\nu}{2}+1;0,\xi+\frac{1}{2},1-\frac{d}{2};\frac{h^2}{a^2}\right)} 
    \text{ if $0 < h < a$}\\
    1 \text{ if $h=0$.}
    \end{cases}
    \end{split}
\end{equation}
\end{enumerate}

\subsection{Proofs}
\label{app:proofs}

\begin{proof}[Proof of Proposition \ref{matgeneral}]

Applying the turning bands operator ${\mathfrak T}_{d+2k,d}$ to the $(d+2k)$-radial Mat\'ern correlation function (\ref{Matern}) gives (Lemma \ref{TB2k})
\begin{equation}
\label{explicitMatern}
  {\cal M}_{a,\xi,d,k}(h) =\sum_{q=0}^k \sum_{r=0}^{\max\{0,q-1\}} \frac{(-1)^r (k-q+1)_q (q)_r (q-r)_r \, h^{q-r} \, {\cal M}_{a,\xi}^{(q-r)}(h)}{2^{q+r} q! \, r! (\frac{d}{2})_q}
\end{equation}
with
\begin{equation*}
{\cal M}_{a,\xi}^{(q-r)}(h) = \frac{2^{1-\xi}}{\Gamma(\xi)} \sum_{s=0}^{q-r} \frac{(q-r)! (\xi+1-s)_{s} \, h^{\xi-s}}{(q-r-s)! s! a^{\xi+q-r-s}}  {\cal K}_{\xi}^{(q-r-s)}\left ( \frac{h}{a} \right )
\end{equation*}
and \citep[10.29.5]{Olver}
\begin{equation*}
  {\cal K}_{\xi}^{(q-r-s)}\left ( \frac{h}{a} \right ) = \left(-\frac{1}{2}\right)^{q-r-s} \sum_{t=0}^{q-r-s} \frac{(q-r-s)!}{t! (q-r-s-t)!} {\cal K}_{\xi+2t+r+s-q}\left ( \frac{h}{a} \right ).
\end{equation*}
These identities lead to (\ref{explicitmatGeneralized}) and prove that the latter is a valid $d$-radial correlation function. Concerning its $d$-radial spectral density, we invoke (\ref{tb2kdensity}) to obtain (\ref{stein2}) from the $(d+2k)$-radial Mat\'ern density (\ref{stein1}).
\end{proof}

\begin{proof}[Proof of Proposition \ref{genW}]

We prove by induction that, under conditions (A) to (C), ${\cal GW}_{a,\xi,\nu,d,k}$ is a valid $d$-radial correlation function and that its $d$-radial spectral density is $\widehat{{\cal GW}}_{a,\xi,\nu,d,k}$.

For $k=0$, the mapping ${\cal GW}_{a,\xi,\nu,d,0}$ defined in (\ref{wendlandext}) coincides with the generalized Wendland correlation (\ref{WG4*}) owing to formula E.2.3 of \cite{matheron1965}. Also, $\widehat{{\cal GW}}_{a,\xi,\nu,d,0}$, as defined by (\ref{GWdensity}), coincides with the $d$-radial spectral density in (\ref{llkk}).

For $k \in \mathbb{N}_{\geq 1}$, assume that ${\cal GW}_{a,\xi,\nu,d+2,k-1}$ is a valid $(d+2)$-radial correlation function under conditions (A) to (C) (these conditions are unchanged when replacing $k$ and $d$ by $k-1$ and $d+2$). The turning bands operator transforms it into a valid $d$-radial correlation, given by (Lemma \ref{TB2})
\begin{equation*}
\label{turning3}
{\mathfrak T}_{d+2,d}\left[{\cal GW}_{a,\xi,\nu,d+2,k-1}\right](h) = \frac{h^{1-d}}{d} \frac{\partial [h^{d} {\cal GW}_{a,\xi,\nu,d+2,k-1}(h)]}{\partial h} , \quad h \geq 0.
\end{equation*}

Using (\ref{GW2Meijer}) and formula 8.2.2.39 of \cite{prud}, one obtains
\begin{equation}
\label{meijer}
    \begin{split}
    &{\mathfrak T}_{d+2,d}\left[{\cal GW}_{a,\xi,\nu,d+2,k-1}\right](h) \\&= \begin{cases}
    0 \text{ if $a \leq h$}\\
    \frac{\Gamma(\frac{d}{2}) \Gamma(\xi+\frac{1+\nu}{2}) \Gamma(\xi+\frac{\nu}{2}+1)}{\Gamma(\xi+\frac{1}{2}) \Gamma(\frac{d}{2}+k)} {G^{2,1}_{3,3}\left(1-\frac{d}{2}-k,\xi+\frac{1+\nu}{2},\xi+\frac{\nu}{2}+1;0,\xi+\frac{1}{2},1-\frac{d}{2};\frac{h^2}{a^2}\right)} 
    \text{ if $0 < h < a$}\\
    1 \text{ if $h=0$}
    \end{cases}\\
    &= {\cal GW}_{a,\xi,\nu,d,k}(h), \quad h \geq 0,
    \end{split}
\end{equation}
which proves that ${\cal GW}_{a,\xi,\nu,d,k} \in \Phi_d$ under conditions (A) to (C). Its $d$-radial spectral density is derived from that of ${\cal GW}_{a,\xi,\nu,d+2,k-1}$ by using (\ref{tb2kdensity}):
\begin{equation*}
  \widehat{{\cal GW}}_{a,\xi,\nu,d,k}(u) = \frac{2\pi u^2}{d} \widehat{{\cal GW}}_{a,\xi,\nu,d+2,k-1}(u), \quad u \geq 0.
\end{equation*}
It is deduced that, if $\widehat{{\cal GW}}_{a,\xi,\nu,d+2,k-1}$ is given by (\ref{GWdensity}), so is $\widehat{{\cal GW}}_{a,\xi,\nu,d,k}$.
\end{proof}

\begin{proof}[Proof of Proposition \ref{WendN}]
Consider a $(d+2)$-radial covariance belonging to the ordinary Wendland class, $i.e.$, $\xi \in \mathbb{N}$ and $k=0$. As shown in \cite{bevi2024}, this covariance can be written as:
\begin{equation}
\label{ordwendform}
{\cal GW}_{a,\xi,\nu,d+2,0}(h) = \sum_{n=0}^{\xi} a_{\xi,n}(\nu) \left(1-\frac{h}{a}\right)_+^{n+\xi+\nu} \left(1+\frac{h}{a}\right)^{\xi-n}
\end{equation}
with  $a_{\xi,n}(\nu) =  \frac{\Gamma(\xi)\Gamma(2\xi+\nu+1) (\nu)_n (-\xi)_n}{2\Gamma(2\xi)\Gamma(\xi+\nu+1) (\xi+\nu+1)_n \, n!}$.
By applying the turning bands operator (\ref{turning03}), one obtains the $d$-radial correlation
\begin{equation*}
\begin{split}
      &{\cal GW}_{a,\xi,\nu,d,1}(h) \\
      & = \begin{cases}
      \sum_{n=0}^{\xi} a_{\xi,n}(\nu) \left(1-\frac{h}{a}\right)^{n+\xi+\nu} \left(1+\frac{h}{a}\right)^{\xi-n}\\
      - \frac{h}{a \,d} \sum_{n=0}^{\xi} a_{\xi,n}(\nu) (n+\xi+\nu) \left(1-\frac{h}{a}\right)^{n+\xi+\nu-1} \left(1+\frac{h}{a}\right)^{\xi-n}\\
      + \frac{h}{a \, d} \sum_{n=0}^{\xi} a_{\xi,n}(\nu) (\xi-n) \left(1-\frac{h}{a}\right)^{n+\xi+\nu} \left(1+\frac{h}{a}\right)^{\xi-n-1}, \quad 0 < h \leq a\\
      0, \quad h \geq a\\
      \end{cases}
\end{split}
\end{equation*}
which coincides with (\ref{genordwend}).

Consider now a $(d+4)$-radial correlation of the form (\ref{ordwendform}) with $\alpha = \frac{d+5}{2}+\xi$, $\xi \in \mathbb{N}$ and $\nu \geq \alpha$. Applying twice the turning bands operator (\ref{turning03}), one obtains
\begin{equation*}
\begin{split}
     &{\cal GW}_{a,\xi,\nu,d,2}(h) =  \sum_{n=0}^{\xi} a_{\xi,n}(\nu) \left(1-\frac{h}{a}\right)_+^{n+\xi+\nu-2} \left(1+\frac{h}{a}\right)^{\xi-n-2} \\
     &\times \Bigg[\left(1-\frac{h^2}{a^2}\right) \left(1 - \frac{(2n+\nu)h}{a \,(d+2)} - \frac{(2\xi+\nu+d+2)h^2}{a^2 \,(d+2)}\right)
       \\
      &- \frac{h(n+\xi+\nu-1)}{a \, d} \left(1+\frac{h}{a}\right) \left(1 - \frac{(2n+\nu)h}{a \,(d+2)} - \frac{(2\xi+\nu+d+2)h^2}{a^2 \,(d+2)}\right)\\
      &+ \frac{h(\xi-n-1)}{a \, d} \left(1-\frac{h}{a}\right) \left(1 - \frac{(2n+\nu)h}{a \,(d+2)} - \frac{(2\xi+\nu+d+2)h^2}{a^2 \,(d+2)}\right)\\
      &+ \frac{h}{d} \left(1-\frac{h^2}{a^2}\right)\left(- \frac{(2n+\nu)}{a \,(d+2)} - \frac{2h(2\xi+\nu+d+2)}{a^2 \,(d+2)}
      \right) \Bigg],\\
\end{split}
\end{equation*}
which yields (\ref{genordwend2}).
\end{proof}

\begin{proof}[Proof of Proposition \ref{wend2mat}]

We first show the convergence of the Generalized Wendland  model to the Mat\'ern (case $k=0$), and will follow with the general case ($k \geq 1$).\\
\textbf{Case 1: $k=0$.}

The pointwise convergence of $\widehat{{\cal GW}}_{\nu a,\xi-{1/2},\nu,d,0}$ to $\widehat{{\cal M}}_{a,\xi,d,0}$ can be established by a straightforward adaptation of the proof given in \cite{bevilacqua2022unifying}. From Lemmas \ref{Dini} and \ref{monotonicity} (owing to (\ref{bevcondition}), the latter lemma applies as soon as $\nu \geq \nu_{\min}(\xi,d+2)$), it is deduced that the convergence to $\widehat{{\cal M}}_{a,\xi}$ is actually uniform on $[0,+\infty)$.

Let us now focus on the convergence of the correlation functions. Before distinguishing different subcases, depending on the value of $\xi$, we recall a useful result: The Mat\'ern correlation with smoothness parameter $\xi > \mu > 0$ is, up to a positive factor, the mont\'ee of order $2\xi-2\mu$ (see Section \ref{app:montee}) of the Mat\'ern correlation with smoothness parameter $\mu$ \citep[II.I.6]{matheron1965}:
\begin{equation}
\label{montMatern}
      {\cal M}_{a,\xi}(h) = \varpi_1(\xi,\mu,a) \int_h^{+\infty} u (u^2-h^2)^{\xi-\mu-1} {\cal M}_{a,\mu}(u) {\rm d}u, \quad h \geq 0,
\end{equation}
with $\varpi_1(\xi,\mu,a) = \frac{2\Gamma(\mu)}{(2 a)^{2\xi-2\mu} \Gamma(\xi)\Gamma(\xi-\mu)}$. In particular, $\varpi_1\left(\xi,\frac{1}{2},a\right) = \frac{1}{a^{2\xi-1}\Gamma(2\xi-1)}$ owing to the gamma duplication formula \citep[5.5.5]{Olver}.\\\\
\textbf{Subcase 1.1: $\xi>\frac{1}{2}$.}

For $b>0$, $\xi>\frac{1}{2}$ and $\nu \geq \xi+\frac{d}{2}$, the mont\'ee of order $2\xi-1$ of the Askey correlation ${\cal GW}_{b,0,\nu,d+\lceil 2\xi-1 \rceil,0}$ is, up to a positive factor, the Generalized Wendland  correlation ${\cal GW}_{b,\xi-{1/2},\nu,d,0}$ \citep[Theorem 9]{emery2021gauss}:
\begin{equation}
\label{montaskey}
      {\cal GW}_{b,\xi-{1/2},\nu,d,0}(h) =  \varpi_2(\xi,\nu,b) \int_h^{+\infty} u (u^2-h^2)^{\xi-{3/2}} {\cal GW}_{b,0,\nu,d+\lceil 2\xi-1 \rceil,0}(u) {\rm d}u, \quad h \geq 0,
\end{equation}
with $\varpi_2(\xi,\nu,b) = \frac{\Gamma(\nu+2\xi)}{b^{2\xi-1} \Gamma(2\xi-1) \Gamma(\nu+1)}$.
Accordingly, for $a,b > 0$, $\xi>\frac{1}{2}$ and $\nu \geq \xi+\frac{d}{2}$,
\begin{equation*}
\begin{split}
    &{\cal M}_{a,\xi}(h)- \frac{\Gamma(\nu+1)}{\Gamma(\nu+2\xi)} \left(\frac{b}{a}\right)^{2\xi-1} {\cal GW}_{b,\xi-{1/2},\nu,d,0}(h) \\
    &= \varpi_1\left(\xi,\frac{1}{2},a\right) \int_h^{+\infty} u (u^2-h^2)^{\xi-{3/2}} \left[{\cal M}_{a,{1/2}}(u) - {\cal GW}_{b,0,\nu,d+\lceil 2\xi-1 \rceil,0}(u)\right] {\rm d}u \\
    &= \varpi_1\left(\xi,\frac{1}{2},a\right) \int_h^{b} u (u^2-h^2)^{\xi-{3/2}} \left[{\cal M}_{a,{1/2}}(u) - {\cal GW}_{b,0,\nu,d+\lceil 2\xi-1 \rceil,0}(u)\right] {\rm d}u \\
    &+ \varpi_1\left(\xi,\frac{1}{2},a\right) \int_b^{+\infty} u (u^2-h^2)^{\xi-{3/2}} {\cal M}_{a,{1/2}}(u) {\rm d}u, \quad h \geq 0.
\end{split}
\end{equation*}

Let $0 < u \leq \nu a$. The following inequalities hold (see \cite[3.6.2]{Mitrinovic1970} for the first two ones; the last inequality stems from the fact that $0 < u \mapsto u^2 \mathrm{e}^{-u}$ is upper bounded by $4 \mathrm{e}^{-2}$ and that $0 < u \mapsto{\cal M}_{a,{3/2}}(u) - \frac{u}{a} \mathrm{e}^{-{u/a}}$ and $0 < u \mapsto {\cal M}_{a,{5/2}}(u) - \frac{u^2}{3a^2} \mathrm{e}^{-{u/a}}$ are non-negative functions, as per Table 1 in \cite{Guttorp}):
\begin{equation}
\label{superinequalities}
    0 \leq \mathrm{e}^{-{u/a}} - \left(1-\frac{u}{\nu a}\right)^\nu \leq \frac{u^2 \, \mathrm{e}^{-{u/a}}}{\nu a^2} \leq \min \left\{\frac{4 \mathrm{e}^{-2}}{\nu}, \frac{u}{\nu a} {\cal M}_{a,{3/2}}(u), \frac{3}{\nu} {\cal M}_{a,{5/2}}(u) \right\}.
\end{equation}

Accordingly, for any $\nu \geq \max\{\xi+\frac{d}{2},\frac{1}{a}\}$ and $h \in (0,\nu a - 1)$, one obtains:
\begin{equation*}
\begin{split}
    0 &\leq {\cal M}_{a,\xi}(h)- \frac{\Gamma(\nu+1) \nu^{2\xi-1} }{\Gamma(\nu+2\xi)} {\cal GW}_{\nu a,\xi-{1/2},\nu,d,0}(h) \\
    &\leq \varpi_1\left(\xi,\frac{1}{2},a\right) \left[\frac{3}{\nu} \int_h^{\nu a} u (u^2-h^2)^{\xi-{3/2}} {\cal M}_{a,{5/2}}(u) {\rm d}u + \int_{\nu a}^{+\infty} u (u^2-h^2)^{\xi-{3/2}} {\cal M}_{a,{1/2}}(u) {\rm d}u\right].
\end{split}
\end{equation*}

On the one hand, owing to (\ref{montMatern}),
\begin{equation*}
     \int_h^{\nu a} u (u^2-h^2)^{\xi-{3/2}} {\cal M}_{a,{5/2}}(u)  {\rm d}u \leq  \frac{{\cal M}_{a,\xi+2}(h)}{\varpi_1\left(\xi+2,\frac{5}{2},a\right)} \leq \frac{1}{\varpi_1\left(\xi+2,\frac{5}{2},a\right)}.
\end{equation*}

On the other hand, based on \cite[3.381.3]{grad},
\begin{equation*}
\begin{split}
    &\int_{\nu a}^{+\infty} u (u^2-h^2)^{\xi-{3/2}} {\cal M}_{a,{1/2}}(u) {\rm d}u \\
    &\leq
    \begin{cases}
        (\nu a)^{2\xi-1} \int_{1}^{+\infty} v^{2\xi-2} \exp\left(-\nu v\right) {\rm d}v = a^{2\xi-1} \Gamma(2\xi-1,\nu) \text{ if $\xi \geq \frac{3}{2}$}\\
        (\nu a)^2 \int_{1}^{+\infty} v \exp\left(-\nu v\right) {\rm d}v = a^{2} \Gamma(2,\nu) \text{ if $0 < \xi < \frac{3}{2}$.}
    \end{cases}
\end{split}
\end{equation*}

Also, for fixed $\xi$, $\frac{\Gamma(\nu+1) \nu^{2\xi-1} }{\Gamma(\nu+2\xi)} \to 1$ as $\nu \to + \infty$ \citep[5.11.12]{Olver}.

It is deduced that, for fixed $a>0$ and $\xi > \frac{1}{2}$, ${\cal GW}_{\nu a,\xi-{1/2},\nu,d,0}$ converges pointwise to ${\cal M}_{a,\xi}$ on $[0,+\infty)$ as $\nu$ tends to infinity. Since ${\cal GW}_{\nu a,\xi-{1/2},\nu,d,0}$ is decreasing and ${\cal M}_{a,\xi}$ is continuous, the convergence is actually uniform on $[0,+\infty)$ owing to Lemma \ref{Dini}.  \\\\
\textbf{Subcase 1.2: $\xi=\frac{1}{2}$.}

The result also holds for $\xi = \frac{1}{2}$. Indeed,
\begin{equation*}
    {\cal GW}_{\nu a,0,\nu,d,0}(h) = \left(1-\frac{h}{\nu a}\right)_+^\nu, \quad h \geq 0.
\end{equation*}

As $\nu$ tends to infinity, the uniform convergence to the exponential correlation ${\cal M}_{a,{1/2}}(h)$ on $[0,+\infty)$ stems from (\ref{superinequalities}).\\\\
\textbf{Subcase 1.3: $\xi<\frac{1}{2}$.}

Let us now consider the case when $0 < \xi < \frac{1}{2}$. Up to a positive factor, the Mat\'ern and Generalized Wendland  correlations ${\cal M}_{a,\xi}$ and ${\cal GW}_{b,\xi-{1/2},\nu,d,0}$ (with $\nu \geq \nu_{\min}(\xi-\frac{1}{2},d)$) are obtained by a descente (mont\'ee of negative order $2\xi-1$, see Section \ref{app:montee}) of the correlations ${\cal M}_{a,{1/2}}$ and ${\cal GW}_{b,0,\nu,d,0}$. This amounts to a descente of order $2$ followed by a mont\'ee of order $2\xi+1$. Accounting for (\ref{montMatern}), (\ref{montaskey}), formula I.4.9 of \cite{matheron1965}, and formulae 3.196.3 and 3.381.4 of \cite{grad} to determine the normalization factors, it comes:
\begin{equation*}
      {\cal M}_{a,\xi}(h) = \frac{1}{a^{2\xi} \Gamma(2\xi)} \int_h^{+\infty} (u^2-h^2)^{\xi-{1/2}} {\cal M}_{a,{1/2}}(u) {\rm d}u, \quad h \geq 0,
\end{equation*}
and
\begin{equation*}
      {\cal GW}_{b,\xi-{1/2},\nu,d,0}(h) =  \frac{\Gamma(\nu+2\xi)}{b^{2\xi} \Gamma(\nu) \Gamma(2\xi)} \int_h^{+\infty} (u^2-h^2)^{\xi-{1/2}} {\cal GW}_{b,0,\nu,d,0}(u) {\rm d}u, \quad h \geq 0.
\end{equation*}

Accordingly, for $a,b > 0$, $0 < \xi < \frac{1}{2}$ and $\nu \geq \nu_{\min}(\xi-\frac{1}{2},d)$,
\begin{equation*}
\begin{split}
    &{\cal M}_{a,\xi}(h)- \frac{\Gamma(\nu)}{\Gamma(\nu+2\xi)} \left(\frac{b}{a}\right)^{2\xi} {\cal GW}_{b,\xi-{1/2},\nu,d,0}(h) \\
    &= \frac{1}{a^{2\xi} \Gamma(2\xi)} \int_h^{+\infty} (u^2-h^2)^{\xi-{1/2}} \left[{\cal M}_{a,{1/2}}(u) - {\cal GW}_{b,0,\nu,d,0}(u)\right] {\rm d}u, \quad h \geq 0.
\end{split}
\end{equation*}

Let $\nu \geq \max\{\nu_{\min}(\xi-\frac{1}{2},d),1+\frac{1}{a}\}$ and $0 < h < \nu a - 1$. Setting $b = \nu a$ and accounting for (\ref{montMatern}) and (\ref{superinequalities}), one gets
\begin{equation*}
\begin{split}
    0 &\leq {\cal M}_{a,\xi}(h)- \frac{\Gamma(\nu) \nu^{2\xi}}{\Gamma(\nu+2\xi)} {\cal GW}_{\nu a,\xi-{1/2},\nu,d,0}(h) \\
    &\leq \frac{1}{a^{2\xi} \Gamma(2\xi)} \left[\frac{1}{\nu a} \int_h^{\nu a} u (u^2-h^2)^{\xi-{1/2}} {\cal M}_{a,{3/2}}(u) {\rm d}u + \int_{\nu a}^{+\infty} {\cal M}_{a,{1/2}}(u) {\rm d}u\right]\\
    &\leq \frac{1}{a^{2\xi} \Gamma(2\xi)} \left[\frac{1}{\nu a} \frac{{\cal M}_{a,\xi+2}(h)}{\varpi_1(\xi+2,\frac{3}{2},a)} + a \exp(-\nu a) \right]\\
    &\leq \frac{1}{a^{2\xi} \Gamma(2\xi)} \left[\frac{1}{\nu a \varpi_1(\xi+2,\frac{3}{2},a)} + a \exp(-\nu a) \right],
\end{split}
\end{equation*}
with $\frac{\Gamma(\nu) \nu^{2\xi}}{\Gamma(\nu+2\xi)}$ tending to $1$ as $\nu$ tends to infinity \citep[5.11.12]{Olver}.

We again invoke Lemma \ref{Dini} to claim that, for fixed $a$ and $\xi$, ${\cal GW}_{\nu a,\xi-{1/2},\nu,d,0}$ uniformly converges to ${\cal M}_{a,\xi}$ on $[0,+\infty)$ as $\nu$ tends to infinity.   \\\\
\textbf{Case 2: $k \geq 1$.}

Let us start with the case $k=1$. By application of the turning bands operator (\ref{turning03}) to the Askey and exponential correlations, one obtains:
\begin{equation*}
{\cal GW}_{\nu a,0,\nu,d,1} = \left(1-\frac{h}{\nu a}\right)_+^{\nu} - \frac{h}{a d} \left(1-\frac{h}{\nu a}\right)_+^{\nu-1}, \quad h \geq 0,
\end{equation*}
and
\begin{equation*}
{\cal M}_{a,{1/2},d,1} = \exp\left(-\frac{h}{a}\right) - \frac{h}{a d} \exp\left(-\frac{h}{a}\right), \quad h \geq 0,
\end{equation*}
with the former correlation tending uniformly to the latter as $\nu$ tends to $+\infty$ owing to (\ref{superinequalities}). The proof that ${\cal GW}_{\nu a,\xi-{1/2},\nu,d,1}$ uniformly converges to ${\cal M}_{a,\xi,d,1}$ as $\nu$ tends to infinity relies on the same arguments as the proof of the case $k=0$. The convergence of ${\cal GW}_{\nu a,\xi-{1/2},\nu,d,k}$ to ${\cal M}_{a,\xi,d,k}$ for any integer $k \geq 2$ can be done similarly by induction.

Concerning the spectral densities, from Equations (\ref{tb2kdensity}), (\ref{stein2}) and (\ref{GWdensity}), one has
$$\widehat{{\cal GW}}_{\nu a,\xi-{1/2},\nu,d,k}(u) = \frac{\pi^k}{(\frac{d}{2})_k} u^{2k} \, \widehat{{\cal GW}}_{\nu a,\xi-{1/2},\nu,d+2k,0}(u), \quad u \geq 0,$$
$$\widehat{{\cal M}}_{a,\xi,d,k}(u) = \frac{\pi^k}{(\frac{d}{2})_k} u^{2k} \, \widehat{{\cal M}}_{a,\xi,d+2k,0}(u), \quad u \geq 0,$$
so that the pointwise convergence of $\widehat{{\cal GW}}_{\nu a,\xi-{1/2},\nu,d,k}$ to $\widehat{{\cal M}}_{a,\xi,d,k}$ as $\nu$ tends to infinity stems from the result established for $k=0$. Owing to Lemmas \ref{Dini} and \ref{monotonicity} (the latter applies as soon as $\nu \geq \nu_{\min}(\xi,d+2k+2)$), it is deduced that the convergence is uniform on $[0,+\infty)$.

\end{proof}

\section*{Acknowledgments}

This work was funded and supported by Agencia Nacional de Investigación y Desarrollo (ANID), under grants ANID FONDECYT 1210050 (X. Emery), ANID AFB230001 (X. Emery), ANID FONDECYT 1240308  (M. Bevilacqua), ANID project Data Observatory Foundation DO210001 (M. Bevilacqua), and MATH-AMSUD1167 22-MATH-06 (AMSUD220041) (M. Bevilacqua). {The authors acknowledge the constructive comments by the associate editor and two anonymous referees, which helped improve the manuscript.}

\bibliography{mybib}
\bibliographystyle{apalike}

\end{document}